\documentclass[preprint,11pt,authoryear]{elsarticle}
\usepackage[british]{babel}
\usepackage{csquotes}

\usepackage{setspace}
\usepackage{titlesec}
\usepackage{amsmath}						
\numberwithin{equation}{section}
\usepackage{graphicx}
\usepackage{boldline,tabu,multirow}
\usepackage{amssymb}
\usepackage{bm}
\usepackage{enumitem}
\usepackage{wrapfig}
\usepackage{float} 
\usepackage{subfig}
\usepackage[table]{xcolor}

\usepackage{graphicx}
\usepackage{amssymb}
\usepackage{amsthm}
\usepackage{mathtools}

\floatstyle{boxed}
\newfloat{program}{thp}{lop}
\floatname{program}{Program}
\usepackage[section]{placeins}
\usepackage{ragged2e}
\usepackage{floatflt} 
\usepackage{array,colortbl,xcolor}
\usepackage[colorlinks=true,linkcolor=blue]{hyperref}
\usepackage{lineno}
\usepackage[linesnumbered,noalgohanging,boxruled,algoruled,vlined]{algorithm2e}
\usepackage{rotating}
\usepackage{bibentry}

\newtheorem{lemma}{Lemma}

\newtheorem{theorem}{Theorem}

\usepackage[labelfont=bf,justification=justified,singlelinecheck=false]{caption}
\usepackage{xcolor}
\definecolor{myblue}{gray}{0.70}

\titleformat{\paragraph} [hang] {\normalfont\normalsize\bfseries} {\theparagraph} {1em} {} 

\definecolor{myblue}{gray}{0.70}
\definecolor{aquamarine}{rgb}{0.5, 1.0, 0.83}
\definecolor{goldenyellow}{rgb}{1.0, 0.87, 0.0}

\usepackage{framed} 

\usepackage{multicol} 

\usepackage{nomencl} 

\makenomenclature

\renewcommand*\nompreamble{\begin{multicols}{2}}
	
\renewcommand*\nompostamble{\end{multicols}}

\journal{*}

\begin{document}
	\onehalfspacing
	
	\begin{frontmatter}
		
		
		\title{Modelling COVID-19 Transmission Dynamics in Ghana}
		
		\author[som,fwwrg,dsas]{Edward Acheampong\corref{cor1}}
		\ead{Edward.Acheampong1@nottingham.ac.uk/eoacheampong@ug.edu.gh}
		\author[dms]{Eric Okyere}
		\author[dsas]{Samuel Iddi}
		\author[dvnmi]{Joseph H. K. Bonney}
		\author[som]{Jonathan A. D. Wattis}
		\ead{Jonathan.Wattis@nottingham.ac.uk}
		\author[fwwrg]{Rachel L. Gomes}
		\ead{Rachel.Gomes@nottingham.ac.uk}
		
		\cortext[cor1]{Corresponding authors}
		
		\address[som]{School of Mathematical Sciences, University of Nottingham, \\ University Park, Nottingham, NG7 2RD, UK}
		\address[fwwrg]{Food Water Waste Research Group, Faculty of Engineering, \\ University of Nottingham, University Park, Nottingham, NG7 2RD, UK}
		\address[dms]{Department of Mathematics and Statistics
			University of Energy \\ and Natural Resources, P.O. Box 214, Sunyani, B/A Ghana}
		\address[dsas]{Department of Statistics and Actuarial Science
			University of Ghana, \\ P.O. Box LG 115, Legon, Ghana}
		\address[dvnmi]{Virology Department, Noguchi Memorial Institute For Medical Research, \\ University of Ghana, P.O. Box LG 581, Legon, Ghana}
		
		\begin{abstract}
			In late 2019, a novel coronavirus, the SARS-CoV-2 outbreak was identified in Wuhan, China and later spread to every corner of the globe. Whilst the number of infection-induced deaths in Ghana, West Africa are minimal when compared with the rest of the world, the impact on the local health service is still significant.  Compartmental models are a useful framework for investigating transmission of diseases in societies. To understand how the infection will spread and how to limit the outbreak. We have developed a modified SEIR compartmental model with nine compartments (CoVCom9) to describe the dynamics of SARS-CoV-2 transmission in Ghana. We have carried out a detailed mathematical analysis of the CoVCom9, including the derivation of the basic reproduction number, $\mathcal{R}_{0}$.  In particular, we have shown that the disease-free equilibrium is globally asymptotically stable when $\mathcal{R}_{0}<1$ via a candidate Lyapunov function. Using the SARS-CoV-2 reported data for confirmed-positive cases and deaths from March 13 to August 10, 2020, we have parametrised the CoVCom9 model. The results of this fit show good agreement with data. We used Latin hypercube sampling-rank correlation coefficient (LHS-PRCC) to investigate the uncertainty and sensitivity of $\mathcal{R}_{0}$ since the results derived are significant in controlling the spread of  SARS-CoV-2. We estimate that over this five month period, the basic reproduction number is given by $\mathcal{R}_{0} = 3.110$,  with the 95\% confidence interval being $2.042 \leq \mathcal{R}_0 \leq 3.240$, and the mean value being $\mathcal{R}_{0}=2.623$.  Of the 32 parameters in the model, we find that just six have a significant influence on $\mathcal{R}_{0}$, these include the rate of testing, where an increasing testing rate contributes to the reduction of $\mathcal{R}_{0}$. 
		\end{abstract}
		
		\begin{keyword}
			Transmission model \sep  SARS-CoV-2 \sep Uncertainty \sep Sensitivity \sep Mathematical analysis \sep Monte Carlo-least squares.
		\end{keyword}
		
	\end{frontmatter}
	
	\section{Introduction}
	\label{S:1}
	
	The recent COVID-19 pandemic has caused a devastating burden on the global economy. Since there are currently no widely-available vaccines to bring down or reduce the infection levels on the susceptible human population, many governmental decision-makers worldwide have resorted to intensive non-pharmaceutical interventions such as wearing of face-masks, social distancing, cleaning of suspected infected surfaces, avoiding crowded places, the use of hand sanitizers. These non-pharmaceutical interventions have significantly helped to reduce the risk of transmission of COVID-19.
	
	Mathematical and statistical modelling tools are important in providing key epidemiological parameters of infectious diseases such as infection or transmission rate, recovery rate, incubation period,  isolation and hospitalization rate, quarantine rate, disease-induced death rate, vaccination rate (with other factors depending on the model formulation)\citep{chowell2009mathematical}. Using mathematical models, parametrised to confirmed reported cases of infection, helps estimate
	the basic reproduction number, $\mathcal{R}_{0}$
	which is a crucial epidemiological parameter that determines whether the infection persists in the  population or dies out \citep{li2020basic, dietz1993estimation, ma2020estimating, roberts2007model,chowell2004basic}.
	
	Nonlinear mathematical models have significantly contributed to the understanding of transmission dynamics of infectious diseases, see, e.g., \citep{hethcote2000mathematics, chowell2016mathematical, brauer2008lecture, khan2015estimating}, and the recent COVID-19 pandemic is of no exception \citep{blyuss2021effects, giordano2020modelling, ali2020role, asamoah2020global, mushayabasa2020role, ndairou2020mathematical, gevertz2020novel, carcione2020simulation}. \cite{qianying2020conceptual} have proposed and studied a data-driven SEIR type epidemic for the recent COVID-19 outbreak in Wuhan which captures the effects of
	governmental actions and individuals' behaviour.
	This literature is growing rapidly;
	\cite{abou2020compartmental} has reviewed
	the fundamentals in SIR/SEIR    modelling of the recent COVID-19 outbreak; here we give a brief overview
	of literature relevant to our work.
	
	\cite{buonomo2020effects} describes a  susceptible-\-infected-\-recovered-\-infected compartmental model to investigate the effects of information-dependent vaccination behavior on COVID-19 infections.
	A simple SEIR COVID-19 epidemic model with nonlinear incidence rates that capture governmental control has been designed by \cite{rohith2020dynamics} to examine the dynamics of the infectious disease in India.
	\cite{pang2020transmission} parametrise a nonlinear SEIHR model to estimate the value and sensitivity of $\mathcal{R}_0$ using data from Wuhan from December 31st, 2019.
	A classic SEIR epidemic is used to study the spreading dynamics of the 2019 coronavirus disease in Indonesia \citep{annas2020stability} using vaccination and isolation as model parameters. They constructed a Lyapunov function to conduct global stability of the disease-free equilibrium point. A data-driven epidemiological model that examines the effect of delay in the diagnosis of the deadly COVID-19 disease is formulated and studied by \cite{rong2020effect}, who estimate parameters and performed a out global sensitivity analysis of their model parameters on $\mathcal{R}_{0}$.
	
	A nonlinear SEIQR COVID-19 epidemic model is introduced by \cite{zeb2020mathematical} who present a local and global stability analysis for their model.  The spread of Covid19 in China due to undetected infections in is examined by \cite{ivorra2020mathematical}.
	\cite{chen2020introduction} propose a model based on dividing the total population  into five non-overlapping classes:  susceptible, exposed, infected (symptomatic infection), asymptomatic infected, and recovered. \cite{sardar2020assessment}, investigate the effects of lockdown using an SEIR model. Using reported cases of this highly infectious disease in some cities and the whole of India, they performed a global sensitivity analysis on the computed $\mathcal{R}_{0}$.
	
	The exposed and infectious epidemiological classes used in formulating infectious diseases models mentioned above have been left as abstract concepts. In reality, especially regarding SARS-CoV-2, it is hard to distinguish between individuals exposed to or infected with SARS-CoV-2,
	due to the presence of asymptomatic carriers. In this present study, we introduce two epidemiological classes, which are: (1) an identified group of exposed individuals suspected to be carriers of SARS-CoV-2 (denoted by $Q$); and, (2) individuals who have been clinically confirmed-positive for SARS-CoV-2 (denoted by $P$). Those identified as suspected exposed individuals are denoted by $Q$ because they are quarantined as required by the COVID-19 protocols in Ghana. Likewise, confirmed-positives ($P$) are treated as infectious individuals who have clinically tested positive for SARS-CoV-2. Introducing these distinctions in the epidemiological classes for SARS-CoV-2 is vital for gaining an understanding of its transmission dynamics within the Ghanaian population.  Using published data from March 13 to August 10, 2020 \citep{ritchie2020coronavirus}, we have parametrised our model using a Monte Carlo-least squares technique together with information derived from literature.
	
	The purpose of this research is to investigate the transmission dynamics of SARS-CoV-2 in Ghana using these more specific epidemiological classes to estimate the basic reproduction number, $\mathcal{R}_{0}$.  We have used Latin-Hypercube Sampling-Partial Rank Correlation Coefficient (LHS-PRCC) technique to quantify the uncertainty in $\mathcal{R}_{0}$ as well as to identify those parameters to which $\mathcal{R}_{0}$ is most sensitive. We have organised the subsequent sections of the paper as follows: in Section \ref{S:2} we present a detailed formulation of an epidemiological model of SARS-CoV-2 transmission in Ghana, together with corresponding mathematical analysis of the positivity and boundedness of solutions, a derivation of the basic reproduction number, and global stability analysis of the disease-free equilibrium, which are given in Section \ref{S:3}. Section \ref{S:4} is dedicated to parameter estimation and numerical simulation. The uncertainty and sensitivity analysis of $\mathcal{R}_{0}$ and its implications are presented in Section \ref{S:5}, together with some simulations predicting possible future dynamics of the epidemic. Finally, we give a brief discussion and conclusion of the study in Section \ref{S:6}.
	
	
	\section{Formulation of the model}
	\label{S:2}
	
	Compartmental models are useful means of qualitatively understanding the dynamics of disease transmissions within a population \citep{martcheva2015introduction, chowell2009mathematical}. In formulating our compartmental model to gain insight into COVID-19 transmission dynamics, the total human population is divided into nine distinct epidemiological classes which are summarised in Table \ref{Ts2a:1}. The numbers of individuals in each category is treated as a continuous variable, the classes being: susceptible, $S(t)$, exposed, $E(t)$, infectious, $I(t)$, quarantined suspects $Q(t)$, confirmed-positive $P(t)$, hospitalised in the ordinary ward $H(t)$, hospitalised in the intensive care unit $C(t)$, self-isolation $F(t)$ and recovered, $R(t)$.
	The total number of individuals in the population is thus given by
	\begin{equation}
		\label{Eqs2:1}
		N(t) = S(t) + E(t) + I(t) + Q(t) + P(t) + H(t) + C(t) + F(t) + R(t) .
	\end{equation}
	
	\begin{table}[H] 
		\hspace*{0.05in}
		\centering
		\begin{minipage}{0.9\textwidth}
			\fontsize{8}{10}\selectfont
			\caption{Description of the variables of the CoVCom9 model.}
			\label{Ts2a:1}
			\centering
			\begin{tabular}{ll}
				\hline \hline
				\multicolumn{1}{l}{Variable}
				& \multicolumn{1}{l}{Description}\\ 	\hline
				$ N $	& Total population\\
				$ S $	& Susceptible individuals\\
				$ E $	& Exposed individuals\\
				$ I $	& Infectious individuals\\
				$ Q $	& Quarantine suspected individuals\\
				$ P $	& Comfirmed-positive individuals\\
				$ H $	& Hospitalised at ordinary ward individuals\\
				$ C $	& Hospitalised at intensive care individuals\\
				$ F $	& Self-isolation individuals\\
				$ R $	& Recovered individuals\\ \hline
			\end{tabular}
		\end{minipage}
	\end{table}
	
	\begin{figure}[H] 
		\hspace*{-0.1in}
		{\includegraphics[scale=0.9]{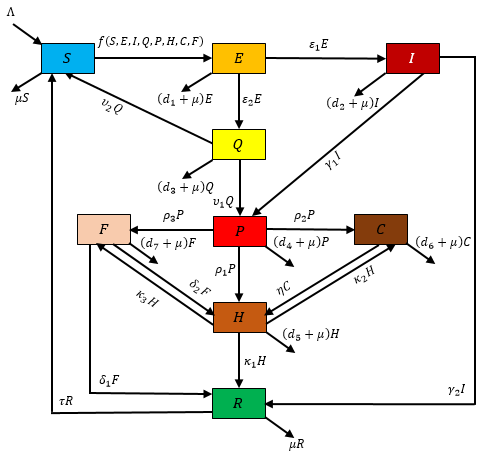}}
		\caption{\textit{Transmission diagram for the model of COVID-19 involving ten compartments. See Tables \ref{Ts2a:1} and \ref{Ts2a:2} for explanations of the parameters and variables used in the model, respectively}.}
		\label{Fs2a:1}
	\end{figure}
	
	Figure \ref{Fs2a:1} summarises the dynamic processes by which individuals pass from one class to another.
	The susceptible class ($S$) represents individuals not exposed to the SARS-CoV-2 virus, and the exposed class ($E$) represents individuals that have recently been exposed to the SARS-CoV-2 virus so are still in the incubation period and can infect others (that is, asymptomatic individuals). An individual in an exposed class can infect another person but with a probability lower than an individual in the infectious class ($I$).
	This rate of infection is given by the nonlinear function $f$ which depends on the parameters $\varphi,\alpha_1,\alpha_2,\alpha_3,\beta_1,\beta_2,\beta_3$.
	Individuals in an infectious class show clear symptoms and have high infectivity. These individuals have not yet been clinically
	confirmed-positive, and thus can spread the disease to the susceptibles.  
	Individuals in class $Q$ are quarantined, that is, individuals identified to have had contact with an exposed individual and so  might be carrying the SARS-CoV-2 virus  (but this has not yet been confirmed), this class also includes individuals not infected with SARS-CoV-2 but are quarantined as a result of enforcement of COVID-19 protocols. These individuals may either enter the susceptible class if test is confirmed negative or to the confirmed-positive class if confirmed to be infected.

	Individuals in the confirmed-positive class $P$ are carriers of the SARS-CoV-2 virus who have had clinical confirmation of this status. These individuals may either enter the intensive care hospitalised class, or be admitted to the ordinary hospitalised class or enter the self-isolated class after this period. The rates of the these processes are governed by the parameters $\gamma_1,v_1,\rho_1,\rho_2,\rho_3$.  
	The individuals in the ordinary Hospitalised class shows some level of sickness due to infection that need to be cared for at the ordinary ward. Though there is chance of entering into recovery class, these individuals' conditions may deteriorate causing them to enter the intensive care hospitalised class. 
	Individuals move between these categories with rates determined by $\kappa_2,\kappa_3,\delta_2,\eta$. These individuals can still infect other individuals who become exposed through close contact. 
	Individuals in intensive care ($C$) can still infect other individuals and have a high risk of dying (rates $d_j$) although improved care conditions may allow transfer to the ordinary ward ($H$, at rate $\eta$).

	Individuals in the self-isolated class ($F$) are on medication at home and can still infect other individuals. These individuals ($F$) may either enter the recovered class ($R$, at rate $\delta_1$) or enter the ordinary hospitalised class (rate $\delta_2$). Individuals who have recovered from SARS-CoV-2 virus enter into the recovered class ($R$) but can be re-infected since there is no life-long immunity, hence there is a flux from $R$ to $S$ with rate parameter $\tau$. We assume that individuals in all the compartments can die of COVID-19 (rates $d_j$) in addition to natural death (rate $\mu$) with the exception of the susceptible compartment with only natural death.  A summary of all the parameter definitions is given in Table \ref{Ts2a:2}.

	\begin{table}[H] 
		\centering
		\begin{minipage}{0.98\textwidth}
			\fontsize{8}{10}\selectfont
			\caption{Description of the CoVCom9 model parameters.}
			\label{Ts2a:2}
			\centering
			\begin{tabular}{p{9mm} p{99mm}} \hline \hline
				\multicolumn{1}{l}{Parameters} & \multicolumn{1}{l}{Description}\\ \hline
				$\Lambda$	& Recruitment rate \\
				$\mu$	    & Natural death rate\\
				$\varphi$   & Transmission rate of infectious individuals ($I$) \\
				$\alpha_1$  & Probability of transmission of exposed individuals ($E$)\\
				$\alpha_2$ & Probability of transmission of quarantine suspected infectious individuals ($Q$)\\
				$\alpha_3$ & Probability of transmission of confirmed-positive infectious individuals ($P$)\\
				$\beta_1$  & Probability of transmission of hospitalised at ordinary ward individuals ($H$)\\
				$\beta_2 $ & Probability of transmission of hospitalised at intensive care individuals ($C$)\\
				$\beta_3$  & Probability of transmission of self-isolation at home individuals ($F$)\\
				$\epsilon_1$ & Progression rate of exposed individuals to infectious class per day\\
				$\epsilon_2$ & Progression rate of exposed individuals to quarantine suspected per day\\
				$\gamma_1$  & Progression rate of infectious individuals to confirmed-positive per day\\
				$\gamma_2$  & Recovery rate of infectious individuals per day\\
				$\upsilon_1$   & Progression rate of quarantined ($Q$) to confirmed cases ($P$) per day\\
				$\upsilon_2$   & Progression rate of quarantined cases to susceptible cases per day\\\
				$\rho_{1}$	& Progression rate of confirmed-positive infectives to hospital class per day\\
				$\rho_{2}$  & Progression rate of confirmed-positive to intensive care class per day\\
				$\rho_{3} $	    & Progression rate of confirmed-positive ($P$) to self-isolation 
				\\ & \quad at home ($F$) class per day \\
				$\kappa_{1}$	& Recovery rate of hospitalised ($H$) individual per day\\
				$\kappa_{2}$	& Progression rate of hospitalised (ordinary, $H$) to intensive care ($C$) per day\\
				$\kappa_{3}$    & Progression rate of hospitalised (ordinary ward, $H$) to\\ & \quad self-isolation at home ($F$) class per day\\
				$ \delta_{1} $	& Recovery rate of self-isolation at home individual per day\\
				$ \delta_{2} $	& Progression rate of self-isolation at home 
				\\ & \quad to hospitalised at ordinary ward ($H$) class per day\\
				$ \eta $	& Progression rate of intensive care to ordinary ward class per day\\
				$ \tau $	& Progression rate of recovery individuals to susceptible class per day\\
				$ d_{1} $	& Disease-induced death rate of exposed individuals per day\\
				$ d_{2} $	& Disease-induced death rate of infectious individuals per day\\
				$ d_{3} $	& Disease-induced death rate of quarantine suspected infectives per day\\
				$ d_{4} $	& Disease-induced death rate of confirmed-positive individuals per day\\
				$ \sigma_{5} $	& Disease-induced death rate of intensive care individuals per day\\
				$ \delta_{6} $	& Disease-induced death rate of self-isolated ($F$) cases per day\\
				$ d_{7} $	& Disease-induced death rate of hospitalised individuals per day\\
				\hline \hline
			\end{tabular}
		\end{minipage}
	\end{table}
	
	The standard form of incidence which is formulated from the basic principles that effective transmission rates are independent of the population size $N$ for human diseases is used in this study \citep{martcheva2015introduction,hethcote2000mathematics}. This principle has been shown in many studies to be a plausible assumption \citep{hethcote2000mathematics}. If $\alpha$ is the average number of sufficient contacts for transmission of an individual per unit time, then $ \alpha I/N $ is the average number of contacts with infectives per unit time of one susceptible, and $ (\alpha I/N) S $ is the incidence. That is, the number of new cases per unit time at time $ t $ due to susceptibles $S(t)$ becoming infected \citep{hethcote2000mathematics}. We use $\varphi$ to denote the effective transmission rate from an infectious individual while $ \alpha_{1} $, $ \alpha_{2} $, $ \alpha_{3} $, $ \beta_{1} $, $ \beta_{2} $ and $ \beta_{3} $ denote the transmission probabilities, of exposed individuals, quarantine suspected exposed individuals, confirmed-positive individuals, ordinary hospitalised individuals, intensive care hospitalised individuals, and self-isolated individual, respectively. All these probabilities lie between zero and one. The incidence is therefore given by
	\begin{equation}  \fontsize{10}{11}\selectfont \label{Eqs2:2}
		f(S,E,I,Q,P,H,C,F) = \varphi \bigg( \frac{\alpha_{1}E+I+\alpha_{2}Q+\alpha_{3}P+\beta_{1}H+\beta_{2}C+\beta_{3}F}{N} \bigg) S .
	\end{equation}
	
	Our COVID-19 model (CoVCom9) is obtained by `translating' the compartmental model summarised in Figure \ref{Fs2a:1} into nine coupled ordinary differential equations
	\begin{subequations}\allowdisplaybreaks	\label{Eqs2:3}
		\begin{align}
			\dfrac{dS}{dt} & = \Lambda + \upsilon_{2} U + \tau R - f(S,E,I,Q,P,H,C,F) - \mu S,
			\label{Eqs2:3a}\\
			\dfrac{dE}{dt} & = f(S,E,I,Q,P,H,C,F)  - (\epsilon_{1} + \epsilon_{2} + \mu + d_{1}) E,
			\label{Eqs2:3b}\\
			\dfrac{dI}{dt} & = \epsilon_{1}E - (\gamma_{1} + \gamma_{2} + \mu + d_{2})I,
			\label{Eqs2:3c}\\
			\dfrac{dQ}{dt} & = \epsilon_{2}E - (\upsilon_{1} + \upsilon_{2} +  \mu + d_{3})Q,
			\label{Eqs2:3d}\\
			\dfrac{dP}{dt} & = \gamma_{1}I + \upsilon_{1}Q - (\rho_{1} + \rho_{2} + \rho_{3} + \mu + d_{4})P,
			\label{Eqs2:3e}\\
			\dfrac{dH}{dt} & = \rho_{1} P + \eta C + \delta_{2}Q - (\kappa_{1} + \kappa_{2} + \kappa_{3}  + \mu + d_{5})H,
			\label{Eqs2:3f}\\
			\dfrac{dC}{dt} & = \rho_{2} P + \kappa_{2} H - (\eta + \mu + d_{6})C,
			\label{Eqs2:3g}\\
			\dfrac{dF}{dt} & = \rho_{3} P + \kappa_{3} H - (\delta_{1} + \delta_{2}  + \mu + d_{7})F,
			\label{Eqs2:3h}\\
			\dfrac{dR}{dt} & = \gamma_{2} I + \kappa_{1} H + \delta_{1}F - (\tau + \mu) R,
		\end{align}
	\end{subequations}
	with $t>0$. These are solved subject to the initial conditions
	\begin{equation}
		\label{Eqs2:4}
		\begin{aligned}
			S(0)=S_{0}\ge 0, \quad E(0)=E_{0}\ge 0, \quad I(0)=I_{0}\ge 0, \\
			Q(0)=Q_{0}\ge 0, \quad P(0)=P_{0}\ge 0, \quad H(0)=H_{0}\ge 0, \\
			C(0)=C_{0}\ge 0, \quad F(0)=F_{0}\ge 0, \quad R(0)=R_{0}\ge 0.
		\end{aligned}
	\end{equation}
	In this paper, we will use the acronym CoVCom9 to indicate the nine compartments of the model of SARS-CoV-2 transmission pattern in Ghana given by Equation (\ref{Eqs2:3}).
	The epidemiologically feasible region of interest of the model (\ref{Eqs2:3}) 
	is the domain defined by 
	\begin{align} 
		\label{Eqs2:5}
		\Omega = \biggl\{(S(t), E(t), I(t), Q(t), P(t), H(t), C(t), F(t), R(t))\in \mathbb{R}_{+}^{9}:
		\notag \\
		S+E+I+Q+P+H+C+F+R \le \frac{\Lambda}{\mu} \biggr\} . 
	\end{align}
	
	In the following sections we present a mathematical analysis of the model 
	with respect to positivity and boundedness of the feasible region, 
	$\Omega$, as well as various stability results and the epidemiological 
	threshold of interest. In the subsequent sections, we discuss a theorem 
	demonstrating that solutions of Equation (\ref{Eqs2:3}) with initial 
	condidtions (\ref{Eqs2:4}) in $\Omega$ remain in $\Omega$.

	
	\section{Mathematical analysis of CoVCom9 model}
	\label{S:3}
	
	\subsection{Positivity, boundedness and invariant region}
	\label{S:3a}
	
	The CoVCom9 model (\ref{Eqs2:3}) depicts COVID-19 transmission dynamics in the human population, so it is vital to show that variables in (\ref{Eqs2:3}) remain  nonnegative and bounded for all time $ t \ge 0 $ and do not leave the
	epidemiologically feasible region of interest, $\Omega$.
	
	\begin{lemma}[Positivity and Boundedness]
		\label{Lm2:1}
		For any given nonnegative initial conditions in Eq. (\ref{Eqs2:4}), the  CoVCom9 model (\ref{Eqs2:3}) has a nonnegative solution
		$\{S(t), E(t), I(t), Q(t), P(t), H(t), C(t), F(t), R(t)\}$ of the system (\ref{Eqs2:3}) for all time $ t \ge 0 $ whenever the parameters are non-negative. Moreover
		\begin{equation}
			\lim_{t\rightarrow \infty} \sup N(t) \le \frac{\Lambda}{\mu} .
		\end{equation}
	\end{lemma}
	
	\begin{proof}
		Considering the first equation of the CoVCom9 model (\ref{Eqs2:3}), one can clearly see that
		\begin{equation}
			\label{Eqs3a:1}
			\dfrac{dS}{dt} \ge - (\lambda + \mu) S,
		\end{equation}
		\noindent where $$ \lambda = \varphi \bigg(\frac{\alpha_{1} E + I + \alpha_{3} Q + \alpha_{4} P + \beta_{1} H + \beta_{2} C + \beta_{3} F}{N}\bigg) $$
		\noindent Next, integrating Eq. (\ref{Eqs3a:1}), we find
		\begin{equation}
			\label{Eqs3a:2}
			S(t) \ge S_{0}\exp\biggl[-\int_{o}^{t}(\lambda(\zeta)+\mu)d\zeta\biggr]
		\end{equation}
		\noindent Therefore $S(t) \ge 0 $ for all $t \ge 0$.
		
		Following a similar argument, it can be shown that the rest of the model variables have nonnegative solutions for all time $t \ge 0 $. That is, $ E(t) \ge 0$, $I(t) \ge 0$, $Q(t) \ge 0$, $P(t) \ge 0$,	$H(t) \ge 0$, $C(t) \ge 0$, $F(t) \ge 0$, $R(t) \ge 0 $, $ \forall t \ge 0 $.
		
		Furthermore, we prove that the solutions are bounded. Adding the right-hand side of the CoVCom9 model (\ref{Eqs2:3}) yields
		\begin{equation}
			\begin{aligned}
				\label{Eqs3a:3}
				\dfrac{dN}{dt} & = \Lambda - \mu N - d E - d_{1} I - d U - d_{2} P - d_{3} H - d_{4} C - d_{5} F \; \le \; \Lambda - \mu N,
			\end{aligned}
		\end{equation}
		\noindent Since $dN/dt  \le \Lambda - \mu N$, it follows that
		\begin{equation}
			\label{Eqs3a:4}
			\lim_{t\rightarrow \infty} \sup N(t) \le \frac{\Lambda}{\mu} .
		\end{equation}
		
	\end{proof}
	
	\begin{lemma}[Positively Invariant Region]
		\label{Lm2:2}
		The region defined by the closed set, $\Omega$ in Eq. (\ref{Eqs2:5}) is positively invariant for the model (\ref{Eqs2:3}) with nonnegative initial conditions (\ref{Eqs2:4}) whenever the parameters are nonnegative.
	\end{lemma}
	
	\begin{proof}
		As in Lemma \ref{Lm2:1}, it follows from the summation of all the equations of the CoVCom9 model (\ref{Eqs2:3}) that
		\begin{equation}
			\begin{aligned}
				\label{Eqs3a:5}
				\dfrac{dN}{dt} & \le \Lambda - \mu N.
			\end{aligned}
		\end{equation}
		\noindent Using the initial condition $N(0)>0$ and an integrating factor, we have
		\begin{equation}
			\begin{aligned}
				\label{Eqs3a:6}
				0 \le N(t) & \le \frac{\Lambda}{\mu} + N(0)\exp(-\mu t),
			\end{aligned}
		\end{equation}
		\noindent where $N(0)$ is the initial value of the total population. Thus $ N(t) \le \Lambda/\mu$, as $ t \rightarrow \infty $. Therefore all feasible solutions of system (\ref{Eqs2:3}) enter the region $\Omega$ defined by (\ref{Eqs2:5}), 
		which is a positively invariant set of the system (\ref{Eqs2:3}). 
		This implies that all solutions in $\Omega$
		remain in $\Omega$  $\forall t \ge 0$. It is therefore sufficient 
		to study the dynamics of CoVCom9 model system (\ref{Eqs2:3}) in $\Omega$.
		
	\end{proof}
	
	\subsection{The basic reproduction number and existence of equilibria}
	\label{S:3b}
	
	The CoVCom9 model has a disease-free equilibrium point given by
	\begin{equation}
		\mathcal{E}_0 = (S_{0}, 0, 0, 0, 0, 0, 0, 0, 0) \in \Omega , 
		\qquad S_{0} = \frac{\Lambda}{\mu} .
	\end{equation}
	The basic reproduction number is defined as the number of secondary infections produced by a single infectious individual during the entire infectious period \citep{vandendriessche2002reproduction}. In this study, the reproduction number defined as the number of secondary SARS-CoV-2 infections generated by a single active SARS-CoV-2 individual during the entire infectious period. Mathematically, the basic reproduction number $\mathcal{R}_0$ is the dominant eigenvalue of the next generation matrix  \citep{Diekmann2010construction,vandendriessche2002reproduction}.
	We apply the method formulation in \citep{vandendriessche2002reproduction} to obtain an expression of $\mathcal{R}_{0}$ for the proposed CoVCom9 (\ref{Eqs2:3}). Let $ {\bf x} = \big(E, I, Q, P, H, C, F\big)^{T}$, then the system (\ref{Eqs2:3}) can be written in the form
	\begin{equation}
		\frac{d \mathbf{x}}{dt} = \mathcal{F}(\mathbf{x}) - \mathcal{V}(\mathbf{x}),
	\end{equation}
	where
	\begin{equation} \label{Eqs3b:1}
		\mathcal{F}(\mathbf{x})=\begin{bmatrix}
			\dfrac{\varphi (\alpha_{1} E + I + \alpha_{2} Q + \alpha_{3} P + \beta_{1} H + \beta_{2} C + \beta_{3} F)S}{N}\\ 0 \\ 0 \\ 0 \\ 0 \\ 0 \\ 0
		\end{bmatrix},
	\end{equation}
	\begin{equation}
		\label{Eqs3b:2}
		\mathcal{V}(\mathbf{x})=\begin{bmatrix}
			\pi_{E} E\\
			-\epsilon_{1} E + \pi_{I} I\\
			-\epsilon_{2} E + \pi_{Q} Q\\
			-\gamma_{1} I-\nu_{1} Q + \pi_{P} P\\
			-\rho_{1} P -\eta C - \delta_{2} F + \pi_{H} H\\
			-\rho_{2} P -\kappa_{2} H + \pi_{C} C\\
			-\rho_{3} P -\kappa_{3} H + \pi_{F} F
		\end{bmatrix}.
	\end{equation}
	and
	\begin{equation}
		\label{Eqs3b:3}
		\begin{aligned}
			\pi_{E} = & \epsilon_{1} + \epsilon_{2} + \mu + d_{1}; \qquad
			\pi_{I}  = \gamma_{1} + \gamma_{2} + \mu + d_{2}; \\
			\pi_{Q} = & \upsilon_{1} + \upsilon_{2} + \mu + d_{3}; \qquad
			\pi_{P}  = \rho_{1} + \rho_{2} + \rho_{3} + \mu + d_{4}; \\
			\pi_{H}  = & \kappa_{1} + \kappa_{2} + \kappa_{3}  + \mu + d_{5}; \qquad
			\pi_{C}  = \eta + \mu + d_{6}; \\
			\pi_{F} = & \delta_{1} + \delta_{2}  + \mu + d_{7} .  	
		\end{aligned}
	\end{equation}
	The Jacobians of $\mathcal{F}(\mathbf{x})$ and $\mathcal{V}(\mathbf{x})$ 
	evaluated at the disease free equilibrium $E_{0}$ are, respectively,
	\begin{equation}
		\label{Eqs3b:4}
		J_F =\begin{bmatrix}
			\varphi \alpha_{1} &\varphi &\varphi \alpha_{2} &\varphi \alpha_{3} &\varphi \beta_{1} &\varphi \beta_{2} &\varphi \beta_{3}\\
			0&0&0&0&0&0&0\\
			0&0&0&0&0&0&0\\
			0&0&0&0&0&0&0\\
			0&0&0&0&0&0&0\\
			0&0&0&0&0&0&0\\
			0&0&0&0&0&0&0\\
		\end{bmatrix},
	\end{equation}
	\begin{equation}
		\label{Eqs3b:5}
		J_V=\begin{bmatrix}
			\pi_{E}& 0&0&0&0&0&0\\
			-\epsilon_{1}& \pi_{I}&0&0&0&0&0\\
			-\epsilon_{2}& 0& \pi_{Q}&0 &0&0&0\\
			0& -\gamma_{1} & -\nu_{1} & \pi_{P}& 0&0&0\\
			0&0&0 &-\rho_{1} &\pi_{H}&-\eta& -\delta_{2}\\
			0 &0& 0& -\rho_{2} & -\kappa_{2} & \pi_{C}& 0\\
			0&0&0&-\rho_{3}&-\kappa_{3} & 0& \pi_{F}
		\end{bmatrix}.
	\end{equation}
	The basic reproduction number, $\mathcal{R}_{0}$ is given by
	the dominant eigenvalue of $J_F J_V^{-1}$
	\begin{align}
		\mathcal{R}_{0}  = & \;\varphi \biggl\{ \dfrac{\alpha_{1}}{\pi_{E}}
		+ \dfrac{\epsilon_{1}}{\pi_{E} \pi_{I}}
		+  \dfrac{\alpha_{2} \epsilon_{2}}{\pi_{E} \pi_Q}
		+ \dfrac{\alpha_{3}}{\pi_P} \biggl( \dfrac{\epsilon_{1} \gamma_{1} }{\pi_{E} \pi_{I}}
		+ \dfrac{\epsilon_{2} \upsilon_{1} }{\pi_{E} \pi_{Q}} \biggr)  
		\nonumber \\
		& + \dfrac{\beta_{1}}{\pi_{P}} \biggl(\dfrac{\rho_{1} \pi_{C} \pi_{F}  + \delta_{2} \rho_{3} \pi_{C} + \eta  \rho_{2} \pi_{F}}{\pi_{H} \pi_{C} \pi_{F} - \delta_{2} \kappa_{3} \pi_{C} - \eta \kappa_{2} \pi_{F}}\biggr) \biggl(\dfrac{\epsilon_{1}}{\pi_{E}}\dfrac{\gamma_{1}}{\pi_{I}}+ \dfrac{\epsilon_{2}}{\pi_{E}}\dfrac{\upsilon_{1}}{\pi_{Q}}\biggr)
		\nonumber \\
		& + \dfrac{\beta_2}{\pi_P} \biggl(\dfrac{\rho_{2} \pi_{H} \pi_{F} + \kappa_{2}  \rho_{1} \pi_{F} + \delta_{2} (\kappa_{2}\rho_{3} \!-\! \kappa_{3}\rho_{2})}{ \pi_{H} \pi_{C} \pi_{F} - \delta_{2} \kappa_{3} \pi_{C} - \eta \kappa_{2} \pi_{F}}\biggr) \biggl(
		\dfrac{\epsilon_{1} \gamma_{1}}{\pi_{E} \pi_{I}} + 
		\dfrac{\epsilon_{2} \upsilon_{1} }{\pi_{E}\pi_{Q}} \biggr)  
		\nonumber \\
		& +\dfrac{\beta_{3}}{\pi_P} \biggl(\dfrac{\rho_{3} \pi_{H} \pi_{C} + \kappa_{3}  \rho_{1} \pi_{C} + \eta (\kappa_{3}\rho_{2} \!-\! \kappa_{2}\rho_{3})}{\pi_{H} \pi_{C} \pi_{F} - \delta_{2} \kappa_{3} \pi_{C} - \eta \kappa_{2} \pi_{F}}\biggr) \biggl( \dfrac{\epsilon_{1}\gamma_{1}}{\pi_{E}\pi_{I}} + \dfrac{\epsilon_{2} \upsilon_{1}}{\pi_{E} \pi_{Q}}\biggr)\biggr\},
		\nonumber \\ & 
		\label{Eqs3b:6}
	\end{align}
	which can be written as
	\begin{align}
		\label{Eqs3b:7}
		\mathcal{R}_{0} = &\; \mathcal{R}_{0E} + \mathcal{R}_{0I} + \mathcal{R}_{0Q} + 
		\mathcal{R}_{0P} + \mathcal{R}_{0H} + \mathcal{R}_{0C} + \mathcal{R}_{0F} ,
	\end{align}
	where the effective reproduction number, $\mathcal{R}_{0}$ is made up of contributions from secondary infections from the exposed group $E$ ($ \mathcal{R}_{0E} $) generated by asymptomatic individuals; the infected (symptomatic) group $I$ ($\mathcal{R}_{0I} $); asymptomatic quarantine suspected individuals - class-$Q$ ($ \mathcal{R}_{0Q} $); confirmed positive individuals - class $P$ ($\mathcal{R}_{0P} $); hospitalised cases ($H$, $ \mathcal{R}_{0H} $); intensive care ($C$) cases, ($ \mathcal{R}_{0C} $); and those self-isolating at home ($F$, $ \mathcal{R}_{0F} $). Equation (\ref{Eqs3b:6}) implies that intervention strategies of SARS-CoV-2 infections should target those in classes $E$, $I$, $Q$, $P$, $H$, $C$, and $F$.
	
	According to Theorem 3.2 of \cite{vandendriessche2002reproduction},  the disease-free steady state $E_{0}$ is locally asymptotically stable if $\mathcal{R}_{0}<1$ and unstable if $\mathcal{R}_{0}>1$.   In the next section we provide stability results for the disease-free equilibrium state.
	
	\subsection{Stability of disease free equilibrium (DFE)}
	\label{S:3c}
	
	In this section, we prove global stability results for the CoVCom9 model (\ref{Eqs2:3}). The epidemiological implication of the local stability is that a small number of the infected individuals will not generate large outbreaks so in the long run, resulting in SARS-CoV-2 dying out provided $ \mathcal{R}_{0} < 1 $. The global stability result helps demonstrate that the disappearance of SARS-CoV-2 disease is independent of the size of the initial subpopulations in the model, provided $ \mathcal{R}_{0} < 1 $ \citep{ali2020role}.  The global stability of the disease-free equilibrium, $E_{0}$ is established using a candidate Lyapunov function.
	
	\begin{theorem}
		\label{Thm3:1}
		The disease-free equilibrium state, $E_{0}$ of the CoVCom9 model (\ref{Eqs2:3}) is globally asymptotically stable in $ \Omega $ if $ \mathcal{R}_{0} < 1 $ and unstable if $ \mathcal{R}_{0} > 1 $.
	\end{theorem}
	
	\begin{proof}
		We construct a candidate Lyapunov function (\ref{Eqs3c:1}) for the CoVCom9 model (\ref{Eqs2:3}) as
		\begin{equation}
			\label{Eqs3c:1}
			V(E,I,Q,P,H,C,F) = \varPhi_{1} E + \varPhi_{2} I + \varPhi_{3} Q + \varPhi_{4} P + \varPhi_{5} H + \varPhi_{6} C + \varPhi_{7} F,	\end{equation}
		where $ \varPhi_{i}, i=1,2,\cdots,7 $ are (as yet unknown) non-negative coefficients.
		Since all the variables are bounded below by zero, then so is $V$.
		Assuming that the variables are solutions of the model (\ref{Eqs2:3}),  the derivative of $V$ with respect to $t$ can be bounded by
		\begin{equation} \allowdisplaybreaks
			\begin{aligned}
				\label{Eqs3c:2}
				\dfrac{dV}{dt}  = & \; \varPhi_{1} \biggl(\varphi \big(\alpha_{1}E+I+\alpha_{2}Q+\alpha_{3}P+\beta_{1}H+\beta_{2}C+\beta_{3}F\big) \big(\frac{S}{N}\big)  - \pi_{E} E \biggr)\\
				& + \varPhi_{2} \biggl(\epsilon_{1}E - \pi_{I} I\biggr) + \varPhi_{3} \biggl(\epsilon_{2}E - \pi_{Q} Q\biggr) + \varPhi_{4} \biggl(\gamma_{1}I + \upsilon_{1}Q - \pi_{P} P\biggr) \\
				& + \varPhi_{5} \biggl(\rho_{1} P + \eta C + \delta_{2}F - \pi_{H} H\biggr) + \varPhi_{6} \biggl(\rho_{2} P + \kappa_{2} H - \pi_{C} C\biggr)  \\
				& + \varPhi_{7} \biggl(\rho_{3} P + \kappa_{3} H - \pi_{F} F\biggr)   \\
				\le & \; \biggl(\varPhi_{1} \varphi \alpha_{1} + \varPhi_{2} \epsilon_{1} + \varPhi_{3} \epsilon_{2} - \varPhi_{1} \pi_{E} \biggr) E +
				\biggl(\varPhi_{1} \varphi + \varPhi_{4} \gamma_{1} - \varPhi_{2} \pi_{I} \biggr) I  \\
				& \!+ \! \biggl(\! \varPhi_{1} \varphi \alpha_{3} \!+\! \varPhi_{4} \upsilon_{1} \!-\! \varPhi_{3} \pi_{Q} \!\biggr) Q  \!+\!
				\biggl(\!\varPhi_{1} \theta \alpha_{4} \!+\! \varPhi_{5} \rho_{1} \!+\! \varPhi_{6} \rho_{2} \!+\! \varPhi_{7} \rho_{3} \!-\! \varPhi_{4} \pi_{P} \!\biggr) P \! \\
				& \!+ \biggl(\varPhi_{1} \varphi \beta_{1} + \varPhi_{6} \kappa_{2} + \varPhi_{7} \kappa_{3} - \varPhi_{5} \pi_{H} \biggr) H +
				\biggl(\varPhi_{1} \varphi \beta_{2} + \varPhi_{5} \eta - \varPhi_{6} \pi_{C} \biggr) C \\
				& \!+ \biggl(\varPhi_{1} \varphi \beta_{3} + \varPhi_{5} \delta_{2} - \varPhi_{7} \pi_{F} \biggr) F, \qquad \qquad \text{since} \qquad S/N<1.\\
			\end{aligned}
		\end{equation}
		Requiring the bracketed coefficients of $E$, $I$, $U$, $P$, $H$, $C$, 
		and $Q$ to zero, we obtain expressions for the previously undetermined 
		parameters $\varPhi_i$, which are thus given by 
		\begin{equation}
			\begin{array}{c}
				\varPhi_{1} = 1, \qquad
				\varPhi_{2} = \dfrac{\varphi + \varPhi_{4}\gamma_{1}}{\pi_{I}}, \qquad
				\varPhi_{3} = \dfrac{\varphi \alpha_{2} + \varPhi_{4}\upsilon_{1}}{\pi_{Q}},
				\\[3ex]
				\varPhi_{4} = \dfrac{1}{\pi_{P}}\biggl[\varphi \alpha_{3}
				+ \dfrac{\varphi \beta_{2} \rho_{2}}{\pi_{C}}
				+ \dfrac{\varphi \beta_{3} \rho_{3}}{\pi_{Q}} + \biggl(\rho_{1}
				+ \dfrac{\eta \rho_{2}}{\pi_{C}}
				+ \dfrac{\delta_{2} \rho_{3}}{\pi_{F}}\biggr)\varPhi_{5}\biggr],
				\\[3ex]
				\varPhi_{5} = \varphi \biggl(\dfrac{\beta_{1}\pi_{C}\pi_{F} + \beta_{2}\kappa_{2}\pi_{F} + \beta_{3}\kappa_{3}\pi_{C}}{\pi_{H}\pi_{C}\pi_{F} - \eta\kappa_{2}\pi_{F} - \delta_{2}\kappa_{3}\pi_{C}}\biggr),
				\\[3ex]
				\varPhi_{6} = \dfrac{\varphi \beta_{2} + \varPhi_{5}\eta}{\pi_{C}},
				\quad \text{and }\quad
				\varPhi_{7} = \dfrac{\varphi \beta_{3} + \varPhi_{5}\delta_{2}}{\pi_{F}} ,
			\end{array}
		\end{equation}
		where the parameter groupings $\pi_*$ are given by (\ref{Eqs3b:3}).
		
		After some simplifications using (\ref{Eqs3b:3}), 
		the time derivative of the Lyapunov function can be written as
		\begin{equation}
			\label{Eqs3c:3}
			\dfrac{dV}{dt} \le \pi_{E} \bigg(\mathcal{R}_{0}-1\bigg) E.	
		\end{equation}
		It is now clear that if $\mathcal{R}_0 < 1$ then $dV/dt\le 0$.
		Furthermore, $dV/dt = 0 $ if $E=0$ and $\mathcal{R}_0<1$. 
		Thus, when $ \mathcal{R}_{0} < 1$, the largest compact invariant set in 
		$\left\{\left( S,E,I,Q,P,H,C,F,R \right)\in\Omega\;|\;\dot{V}\le 0\right\}$ 
		is the single state $\mathcal{E}_0$. 
		LaSalle’s Invariance Principle then implies that $\mathcal{E}_{0}$ 
		is globally asymptotically stable in $\Omega$ if $\mathcal{R}_{0}<1$. 
	\end{proof}
	
	
	\section{CoVCom9 model estimation and numerical simulations}
	\label{S:4}
	
	\subsection{Methodology}
	
	In this section, we briefly describe the parameter estimation and numerical simulation process used to investigate how well the proposed CoVCom9 model (\ref{Eqs2:3}) agrees with the confirmed cases and deaths in Ghana. Here, we consider the SARS-CoV-2 confirmed cases and deaths from March 13, 2020 to August 10, 2020 as reported in Ghana. The data are obtained from Our World in Data \citep{ritchie2020coronavirus}.
	
	The CoVCom9 model (\ref{Eqs2:3}) has nine state variables; to obtain the disease-induced mortality ($D$), we introduce the extra equation
	\begin{equation} \allowdisplaybreaks \label{Eqs4:1}
		\dfrac{dD}{dt} = d_{1} E + d_{2} I + d_{3} Q + d_{4} P + d_{5} H + d_{6} C + d_{7} F,
	\end{equation}
	which introduces no additional parameters.
	The CoVCom9 model has a total of 35 parameters to estimate using limited data (confirmed-positive cases and deaths only). This results in identifiability issues causing the non-convergence of the optimisation of the objective function. We implement the following practical principles to choose reasonable initial parameter values:
	\begin{enumerate}
		
		\item Expert review process which involves asking health experts and/or consulting the relevant literature as well as individuals' experience of the infection. Accordingly, an estimate of the model parameters, natural birth rate, $\mu$, recruitment rate, $ \Lambda$, incubation period, $ \epsilon_{1} $, and recovery rate of quarantine/self-isolation at home individual, $ \delta_{1}$ are obtained. We assumed that the life expectancy of people in Ghana is estimated as 64.35 years \citep{asamoah2020global}, then the natural death rate is estimated as $ \mu = 1/(64.35 \times 365) \approx 4.258 \times 10^{-5} $ per day.
		The population of Ghana in 2020 is estimated to be $ N = 30,960,000$ \citep{gss2020}, and the recruitment rate of humans is estimated as $\Lambda = \mu N \approx 1.318 \times 10^{3}$ people per day.
		The incubation period is 3–7 days, here we choose $ \epsilon_{1} = 1/5.88$ per day as estimated by \cite{pang2020transmission} which is consistent with the wider literature \citep{anderson2020will,li2020early}. The self-isolated positive-confirmed individuals on medication take 14 days on average to recover, thus we assume $\delta_{1} = 1/14$ per day.
		
		\item Exploring the model using the available data (also known as `system exploratory analysis' (SEA) \citep{sms2018}). This process helps identify ranges of parameter values where the trajectories of the CoVCom9 are consistent with the data, and regions of parameter space where trajectories deviate from the times series data of confirmed-positive cases and deaths.
		The motivation for this approach is to restrict the ranges of the parameters and so reduce risk of the Monte Carlo simulation getting trapped at a local optima. Since we have 31 remaining model parameters to infer, applying this SEA technique yields upper and lower  bounds for the model parameters which are presented in (Table \ref{Ts7:2}).
	\end{enumerate}
	
	We use a Monte Carlo least squares method to infer model parameter since it is reliable and efficient. This method seeks to generate the best Monte Carlo estimate ($\widehat \theta_j$) of the model parameters ($\theta$, listed in Table \ref{Ts2a:2}) by minimising the error between the observed data (confirmed-positive cases and deaths), $\boldsymbol{Y}_{j}$ and the simulated data from the CoVCom9 model (\ref{Eqs2:3}), $ \boldsymbol{Y}_{j}^{sim}$ given by the variables listed in Table \ref{Ts2a:1}. 
	Denoting the total number of data-points by $n$ and using $i$ 
	($1 \leq i \leq M$) to enumerate the Monte Carlo simulations, we have 
	\begin{equation}
		\allowdisplaybreaks
		\label{Eqs4:2}
		\widehat{\boldsymbol{\theta}}_{j}^{(i)} = \arg \min_{\theta}\sum_{j=1}^{n} \biggl(\boldsymbol{Y}_{j}-\boldsymbol{Y}_{j}^{sim}\biggr)^{2}, \quad (i=1,2,3,\cdots,M) . 
	\end{equation}
	
	Finally, for the $M$ Monte Carlo samples of $ \widehat{\boldsymbol{\theta}} $, 
	we obtain the mean and covariance matrix of the estimator, $ \widehat{\theta}_{M} $ of $ \boldsymbol{\theta} $ as
	\begin{equation}
		\label{Eqs4:3}
		\widehat{\boldsymbol{\theta}}_{M} = \frac{1}{M} \sum_{i=1}^{M}\widehat{\boldsymbol{\theta}}^{(i)},
	\end{equation}
	
	\begin{equation}
		\label{Eqs4:4}
		\widehat{\boldsymbol{\Sigma}}_{M}=\frac{1}{M-1} \sum_{i=1}^{M}\left(\widehat{\boldsymbol{\theta}}^{(i)}-\widehat{\boldsymbol{\theta}}_{M}\right)\left(\widehat{\boldsymbol{\theta}}^{(i)}-\widehat{\boldsymbol{\theta}}_{M}\right)^{T}.
	\end{equation}
	We also give a 95\% confidence interval of the Monte Carlo 
	samples $\{\widehat{\boldsymbol{\theta}}^{(i)}\}_{i=1}^{M}$ as
	\begin{equation}
		\label{Eqs4:5}
		\biggl(\widehat{\boldsymbol{\theta}}_{M}^{*(0.025)} \, , \;
		\widehat{\boldsymbol{\theta}}_{M}^{*(0.975)}\biggr),
	\end{equation}
	where $\widehat{\boldsymbol{\theta}}_{M}^{*(0.025)}$ and
	$\widehat{\boldsymbol{\theta}}_{M}^{*(0.975)}$ are respectively the
	$\widehat{\boldsymbol{\theta}}^{*(i)}$ in the 2.5\% and 97.5\% positions of
	the ordered Monte Carlo samples $\{\widehat{\boldsymbol{\theta}}^{*(i)}\}_{i=1}^M$.
	
	During parameter estimation, we use a logarithmically transformed
	parameter vector, $\log\boldsymbol{\theta}$, since:
	(i) this conveniently ensures that all parameters are positive, $\theta>0$; and
	(ii) this improves the numerical search of the parameter space across a wide
	range of $\boldsymbol{\theta}$ \citep{bland1996statistics,acheampong2019modelling}.
	All computations use MATLAB, 2018a  .
	
	
	\subsection{Results of CoVCom9 model parameter estimation}
	\label{S:4a}
	
	\begin{table}[H] 
		\hspace*{0.05in}
		\centering
		\begin{minipage}{0.9\textwidth}
			\fontsize{8}{10}\selectfont
			\caption{Estimated initial values of model variables for the system (\ref{Eqs2:2} using Monte Carlo least squares (MC-LS) method). }
			\label{Ts4a:1}
			\centering
			\begin{tabular}{p{1cm} p{23mm} p{35mm} p{19mm}}
				\hline 
				\multicolumn{1}{l}{Variables}
				& \multicolumn{1}{l}{Initial values}
				& \multicolumn{1}{l}{95\% Confidence Interval }
				& \multicolumn{1}{l}{Reference}\\  \hline \hline 
				$ N $ & 30,955,202 & & \cite{gss2020}\\
				$ S $ & 30,954,982 & & \\
				$ E $ & 214.0 & 79.49 - 261.3 & MC-LS\\ \hline 
				$ I $ & 0.346689 & 0.1959 - 3.579033 & MC-LS\\
				$ Q $ & 2.932 & 1.732 - 11.85 & MC-LS\\
				$ P $ & 2 & & \cite{ritchie2020coronavirus}\\ \hline 
				$ H $ & 0 & & \cite{ritchie2020coronavirus}\\
				$ C $ & 0 & & \cite{ritchie2020coronavirus}\\
				$ F $ & 0 & & \cite{ritchie2020coronavirus}\\
				$ R $ & 0 & & \cite{ritchie2020coronavirus}\\
				\hline 
			\end{tabular}
		\end{minipage}
	\end{table}

	In this section, the results obtained using the Monte Carlo least-squares technique described in Section \ref{S:4} are presented.
	Table \ref{Ts4a:1} shows initial values of the state variables;  those for compartments $E$, $I$ and $Q$ are estimated from the reported data.
	From Table \ref{Ts4a:1}, we infer that while on March 13, 2020 two individuals are reported to be confirmed-positive of SARS-CoV-2 infection, the corresponding number of individuals in the exposed ($E$), infectious ($I$) and quarantine-suspected ($Q$) compartments are approximately 213, 1, and 3 respectively.
	
	Table \ref{Ts4a:2} gives the parameter values obtained together with their confidence intervals.
	We note that the infectivity of the individuals in the infected compartment ($I$) is stronger than the other compartments: in decreasing order, the infectivities are due to the groups $E$, $F$, $Q$, $H$, $C$, and $P$.
	The overall transmission rate of the SARS-CoV-2 infection in Ghana for the duration of the data considered in this study is $\varphi=$0.02495 per day, .

	\begin{table}[H] 
		\centering
		\begin{minipage}{0.97\textwidth}
			\fontsize{8}{10}\selectfont
			\caption{Estimated values of the model parameters for the system (\ref{Eqs2:2}) using Monte Carlo least squares (MC-LS) method.}
			\label{Ts4a:2}
			\centering
			\begin{tabular}{p{10mm} p{22mm} p{43mm} p{27mm}}
				\hline 
				\multicolumn{1}{l}{Parameter}
				& \multicolumn{1}{l}{Value}
				& \multicolumn{1}{l}{95\% Confidence Interval}
				& \multicolumn{1}{l}{Reference}\\ \hline \hline 
				$ \alpha_{1} $ & 0.8412
				& (0.2981 $,\;$ 1.0000) & MC-LS\\
				$ \varphi $ & 0.02495
				& (0.02300 $,\;$ 0.03969) & MC-LS\\
				$ \alpha_{2} $ & 0.3152
				& (0.2106 $,\;$ 0.7553) & MC-LS\\ \hline
				$ \alpha_{3} $ & 0.05744
				& (0.02168 $,\;$ 0.08098) & MC-LS\\
				$ \beta_{1} $ & 0.2606
				& (0.09697 $,\;$ 0.3576) & MC-LS\\
				$ \beta_{2} $ & 0.1205
				& (0.06108 $,\;$ 0.2436) & MC-LS\\ \hline
				$ \beta_{3} $ & 0.4857
				& (0.1787 $,\;$ 0.6772) & MC-LS\\
				$ \epsilon_{1} $ & 1/5.882
				& (1/7 $,\;$ 1/3) & \cite{zhang2020novel}\\
				$ \epsilon_{2} $ & 0.001144
				& (0.000873 $,\;$ 0.003217) & MC-LS\\ \hline
				$ \gamma_{1} $ & 0.01004
				& (0.008402 $,\;$ 0.02817) & MC-LS\\
				$ \gamma_{2} $ & 0.000163
				& (8.663$\times 10^{-5}$ $,\;$ 3.300$\times 10^{-4}$) & MC-LS\\
				$ \upsilon_{1} $ & 0.000524
				& (0.000379 $,\;$ 0.001293) & MC-LS\\ \hline
				$ \upsilon_{2} $ & 1.418$\times 10^{-9}$
				& (6.864$\times 10^{-10}$ $,\;$ 2.745$\times 10^{-9}$) & MC-LS\\
				$ \rho_{1} $ & 0.001971
				& (0.000749 $,\;$ 0.002364) & MC-LS\\
				$ \rho_{2} $ & 5.075$\times 10^{-6}$
				& (2.494$\times 10^{-6}$ $,\;$ 9.911$\times 10^{-6}$) & MC-LS\\ \hline
				$ \rho_{3} $ & 0.004711
				& (0.001950 $,\;$ 0.005565) & MC-LS\\
				$ \kappa_{1} $ & 0.008619
				& (0.005728 $,\;$ 0.02076) & MC-LS\\
				$ \kappa_{2} $ & 5.844$\times 10^{-6}$
				& (2.874$\times 10^{-6}$ $,\;$ 1.150$\times 10^{-5}$) & MC-LS\\ \hline
				$ \kappa_{3} $ & 3.009$\times 10^{-5}$
				& (1.488$\times 10^{-5}$ $,\;$ 5.949$\times 10^{-5}$) & MC-LS\\
				$ \delta_{1} $ & 1/14
				& (1/23 $,\;$ 1/11) & Assumed\\
				$ \delta_{2} $ & 5.865$\times 10^{-9}$
				& (2.863$\times 10^{-9}$ $,\;$ 1.145$\times 10^{-8}$) & MC-LS\\ \hline
				$ \eta $ & 9.771$\times 10^{-5}$
				& (4.862$\times 10^{-5}$ $,\;$ 0.000194) & MC-LS\\
				$ \tau $ & 1.538$\times 10^{-8}$
				& (7.741$\times 10^{-9}$ $,\;$ 3.096$\times 10^{-8}$) & MC-LS\\
				$ d_{1} $ & 7.780$\times 10^{-10}$
				& (3.893$\times 10^{-10}$ $,\;$ 1.557$\times 10^{-9}$) & MC-LS\\ \hline
				$ d_{2} $ & 1.249$\times 10^{-13}$
				& (6.222$\times 10^{-14}$ $,\;$ 2.488$\times 10^{-13}$) & MC-LS\\
				$ d_{3} $ & 0.002877
				& (0.001032 $,\;$ 0.003985) & MC-LS\\
				$ d_{4} $ & 6.004$\times 10^{-10}$
				& (2.997$\times 10^{-10}$ $,\;$ 1.199$\times 10^{-9}$) & MC-LS\\ \hline
				$ d_{5} $ & 1.392$\times 10^{-12}$
				& (6.839$\times 10^{-13}$ $,\;$ 2.735$\times 10^{-12}$) & MC-LS\\
				$ d_{6} $ & 6.967$\times 10^{-14}$
				& (3.413$\times 10^{-14}$ $,\;$ 1.365$\times 10^{-13}$) & MC-LS\\
				$ d_{7} $ & 2.455$\times10^{-12}$
				& (1.201$\times10^{-12}$ $,\;$ 4.804$\times10^{-12}$) & MC-LS\\
				\hline 
			\end{tabular}
		\end{minipage}
	\end{table}

	The corresponding best fits of the model to the reported data and the two-year simulations based on the estimated parameter estimates are shown in Figure \ref{Fs4a:1}.
	
	\begin{figure}[H] 
		\hspace*{-0.55in}
		{\includegraphics[scale=0.27]{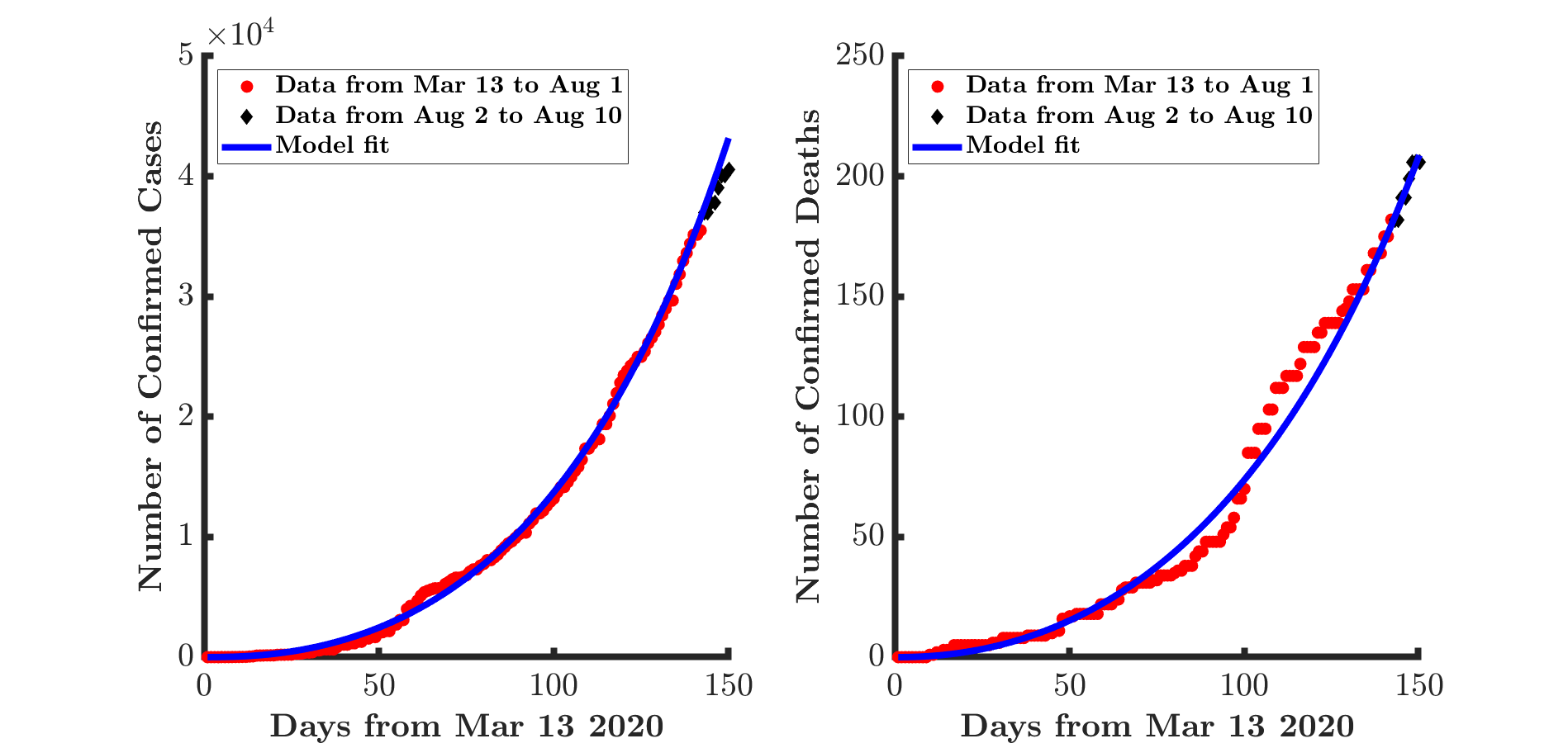}}
		\caption{\textit{Dynamics of CoVCom9 model showing model fit (\textbf{blue line}) and reported data (\textbf{red and black dots}) for (\textbf{Left panel}) daily numbers of confirmed cases simulated from the CoVCom9 model and the numbers from the report data (\textbf{Left panel}) daily numbers of confirmed deaths simulated from the CoVCom9 model and the numbers from the report data from March 13, 2020 to August 10, 2020}.}
		\label{Fs4a:1}
	\end{figure}
	
	The rate at which individuals transfer from Classes $E$ to $Q$ is $\epsilon_2=0.001144$, indicating that each day only 0.11\% of the individuals exposed to the SARS-CoV-2 infection are identified with the suspicion of carrying the infection, and can be contact traced and quarantine in order to clinically confirm their status as either positive or negative of the infection. The rate at which these suspected exposed individuals are confirmed positive is estimated to be $\upsilon_1=0.000524$. For individuals confirmed-positive with the SARS-CoV-2 infection, we can infer from Table \ref{Ts4a:2} that the rate at which individuals progress to Intensive Care (Class $C$) is low compared to the rate at which they progress to either $H$ or $F$ Classes (standard hospital ward or self-isolating at home), with the rate of progression from $P$ to $F$ Classes the highest (that is from positive test to home isolation).  The recovery rate of individuals in Class $H$ is estimated as $\kappa_1=0.008619$ and the rate at which these individuals losing immunity and becoming susceptible to the SARS-CoV-2 infection is $\tau=1.538\times 10^{-8}$; indicating that the rate of SARS-CoV-2 re-infection in Ghana is extremely low (full details of parameters and ranges is given in Table \ref{Ts4a:2}).
	
	From equation (\ref{Eqs3b:6}) and the parameter estimates in 
	Table \ref{Ts4a:2}, the basic reproduction number, $\mathcal{R}_0$,  
	is estimated to be 3.110. The breakdown of this estimate is given, 
	in decreasing order, by
	\begin{itemize}
		\item primarily, symptomatic individuals (class $I$, giving 
		$\mathcal{R}_{0I} = 2.417$),
		\item hospitalised cases (class $H$, contributing 
		$\mathcal{R}_{0H} = 0.212$),
		\item positively tested individuals (class $P$ giving 
		$\mathcal{R}_{0P} = 0.207$),
		\item infections due asymptomatic cases (class $E$, giving         
		$\mathcal{R}_{0E}=0.123$),
		\item self-isolating individuals (class $F$ contributing 
		$\mathcal{R}_{0F} = 0.116 $),
		\item intensive care cases (class $C$, giving $\mathcal{R}_{0C} = 0.020$),
		\item those quarantined at home (class $Q$ contributing 
		$\mathcal{R}_{0Q}=0.015$).
	\end{itemize}
	The basic reproduction number of COVID-19 based on the proposed CoVCom9 model 
	for Ghana is higher than that of many other countries, which indicates a 
	greater epidemic risk in Ghana. A recent study by \cite{asamoah2020global} 
	provides a similar estimate of $\mathcal{R}_0$ in Ghana of $2.64$, 
	differs by only  15\% from our estimate.   
	However, the number of deaths reported in Ghana is low compared to 
	that of other countries in the world.
	For published values for other  countries, 
	please see 
	\cite{ali2020role,zeb2020mathematical,brandi2020epidemic,mwalili2020seir,kumar2020data,mushayabasa2020role,gotz2020early,chen2020introduction,sardar2020assessment,ivorra2020mathematical,asamoah2020global,rahman2020basic}. 
	
	Using the estimated parameter values given in Tables \ref{Ts4a:1} and \ref{Ts4a:2}, the one-year simulation transmission dynamics of the CoVCom9 model offers insight into the SARS-CoV-2 among Ghanaian with respect to the COVID-19 protocols which are in place in the country. Figure \ref{Fs4a:2} depicts the one-year simulation dynamics for the classes $E$, $I$, $Q$, $P$, $H$, $C$, $F$, and deaths ($D$).
	
	\begin{figure}[H] 
		\hspace*{-0.60in}
		\centering
		\includegraphics[scale=0.28]{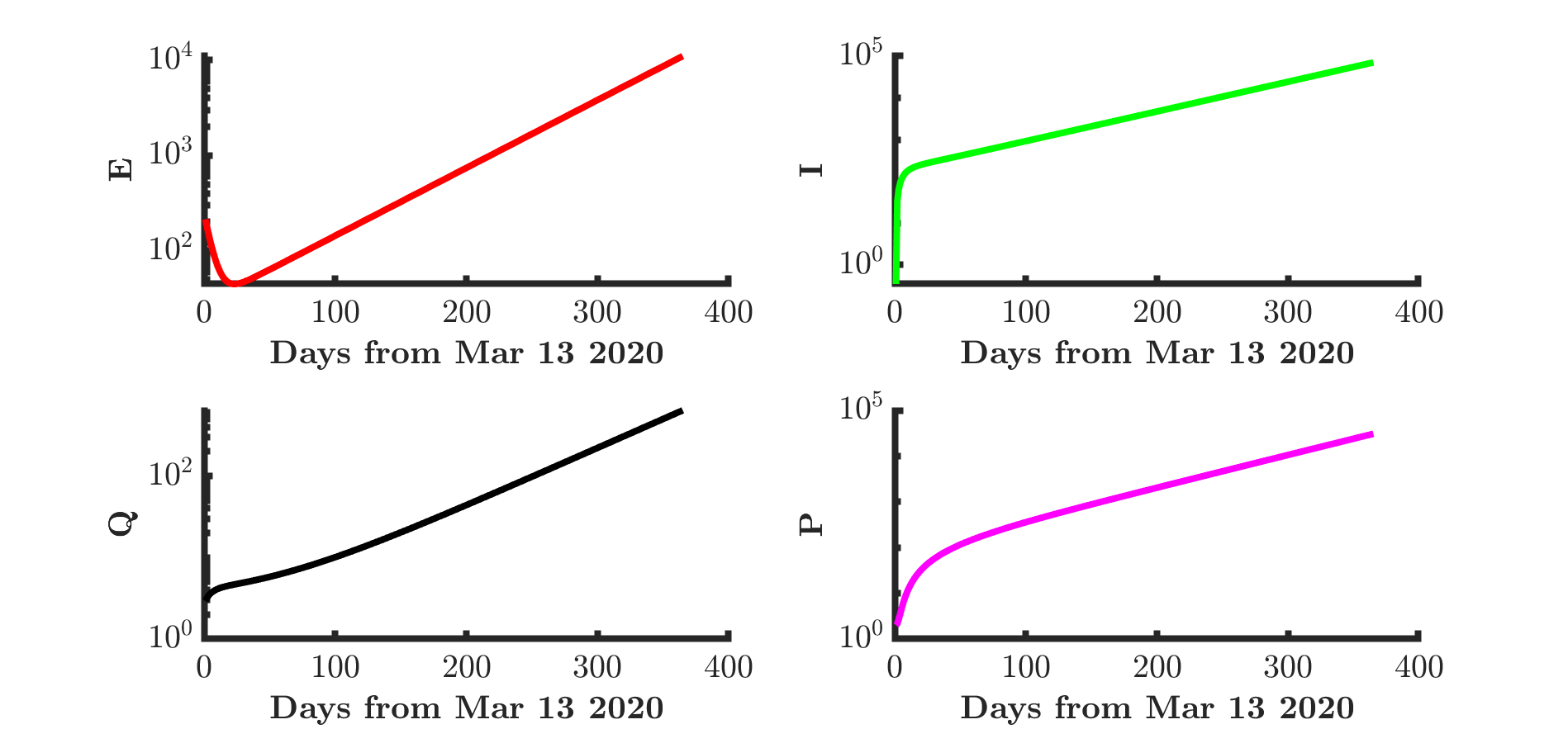}\\
		\hspace*{-0.6in}
		\includegraphics[scale=0.28]{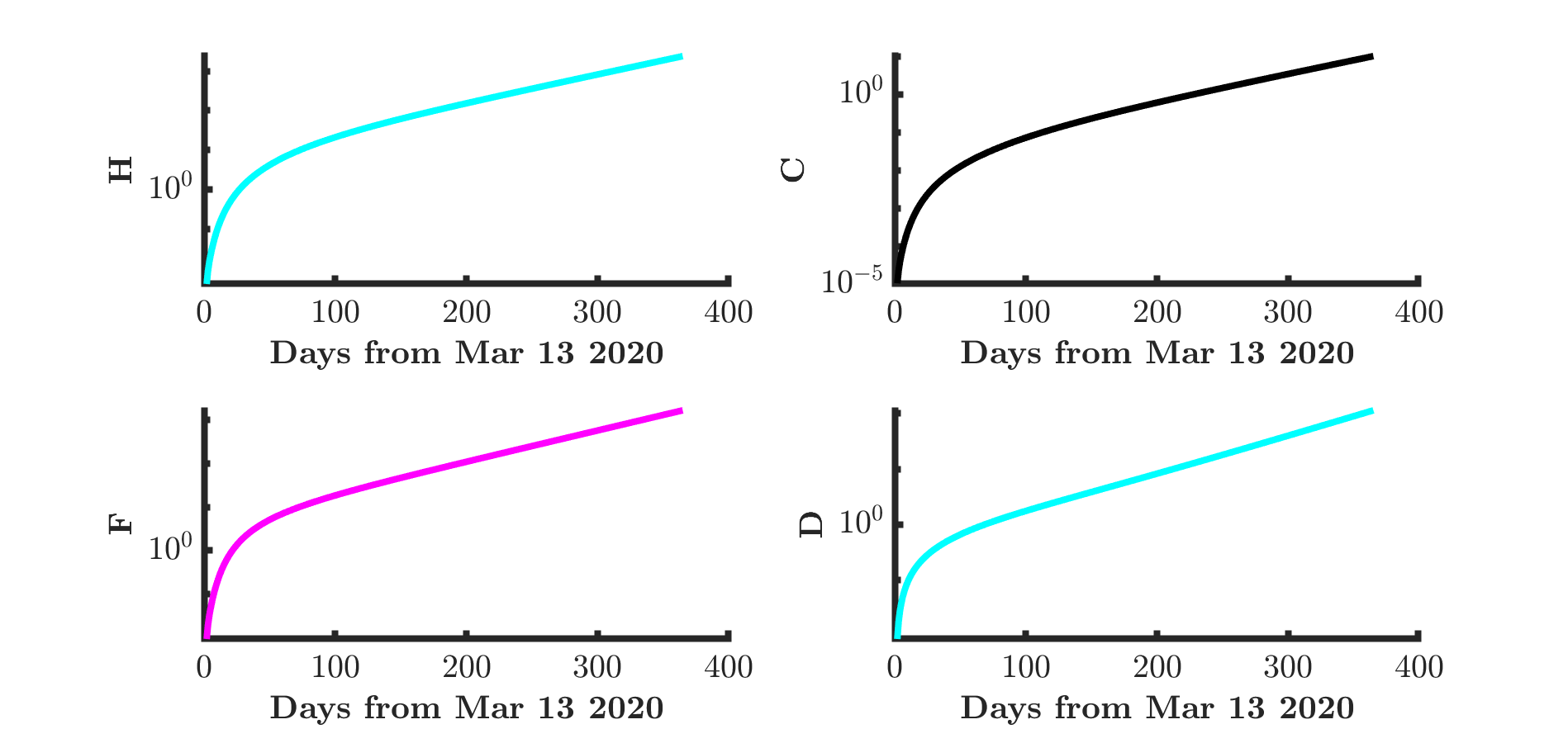}
		\caption{\textit{One-year simulation dynamics of CoVCom9 model from March 13 2020 where $E$, $I$, $Q$, $P$, $H$, $C$, $F$ and $D$ are respectively exposed, infectious, quarantine suspected-expose, confirmed-positive, hospitalised at ordinary ward, hospitalised at intensive care unit, and deaths with the vertical axis on a log-scale}.}
		\label{Fs4a:2}
	\end{figure}
	
	As shown in Figure \ref{Fs4a:2}, all state variables in the CoVCom9 model show an increasing trend, indicating that Ghana continuing the same protocols may not be enough to eradicate the SARS-CoV-2 infection. This has been further complicated by theopening of the the borders, 
	meaning that new control measures are needed to mitigate the spread (both in and out).  Our projections show that with Ghana exercising current  COVID-19 protocols the actual cases substantially exceed those reported (whether hospitalised or only positively tested).
	We thus expect the exponential growth to continue. 
	
	In the next section we discuss the derivation of the basic reproduction number from the CoVCom9 model, and identify influential parameters that intervention strategies should focus on in order to control the spread of the virus.
	
	
	\section{Uncertainty and sensitivity analysis of the basic reproduction number}
	\label{S:5}
	
	\subsection{Methodology}
	
	The proposed CoVCom9 model (\ref{Eqs2:3}) has many unknown parameters. Due to the limited data available, there is substantial uncertainty in calibrating the values of the 31 CoVCom9 model (\ref{Eqs2:3}) parameters \citep{marino2008methodology}.  However, in all cases the ratio of the upper bounds of the 95 \% confidence interval is less than five times the lower bound, and more often four or below,  thus so the order of magnitude of all parameters is well established.
	Since the intervals are derived using the logarithm of parameter values, and our best estimates lie in the centre of this band, each upper bound is approximately twice the estimate and the lower bound half of it.
	This uncertainty in model parameters results in some variability in the prediction of the basic reproduction number $ \mathcal{R}_{0}$.  Latin Hypercube Sampling-Partial Rank Correlation Coefficient (LHS-PRCC) sensitivity analysis was used to evaluate variabilities in the basic reproduction number $ \mathcal{R}_{0} $. The LHS-PRCC approach provides an opportunity to examine the entire parameter space of the CoVCom9 model (\ref{Eqs2:3}) with computer simulations.
	
	We analyse the impacts of the LHS parameters on the basic reproduction number $ \mathcal{R}_{0} $ of the CoVCom9 model (\ref{Eqs2:3}) via standard Monte Carlo procedure. The key parameters to which $\mathcal{R}_{0}$, given by (\ref{Eqs3b:6}), is most sensitive are determined using the PRCCs values, suggesting the most effective way of controlling SARS-CoV-2 infection. Moreover, this analysis also identifies which parameters need to be known precisely when estimating $\mathcal{R}_{0}$ from data \citep{marino2008methodology}.
	
	The application of the combined LHS-PRCC methodology in infectious disease modelling are  fully described elsewhere, for example, in \citep{wu2013sensitivity,marino2008methodology}. This method generally involves: 
	\begin{description}
		\item[(i)] generating LHS parameters in matrix form, together with a ranking of outcome measures $\mathcal{R}_{0}$;  
		\item[(ii)] construction of two linear regression models in response to each parameter and outcome measure, and 
		\item[(iii)] computation of a Pearson rank correlation coefficient for the residuals from the two regression models to obtain the PRCC values for that particular parameter \citep{marino2008methodology,orwa2019uncertainty}.  
	\end{description}
	We induce the correlation between the input parameters using the rank-based method of Iman and Conover \citep{iman1982distribution}. The correlation matrix for the 28 model parameters (listed in Table \ref{Ts2a:2}) is obtained from the parameter estimation in Section \ref{S:4}, where no correlation is assumed between the parameters $ \epsilon_{1} $ and $\delta_{1} $ and other parameters, since these two parameters are not included in the parameter estimation.
	
	
	\subsection{Results of Analysing the LHS-PRCC for the CoVCom9 model}
	\label{S:5a}
	
	The result of the uncertainty analysis of the basic reproduction number, $\mathcal{R}_{0}$ (\ref{Eqs3b:6}) of the CoVCom9 obtained by generating 1000 LHS samples using the Monte Carlo technique is presented in Figure \ref{Fs5a:1}. This histogram depicts the uncertainty in $\mathcal{R}_{0}$, where the degree of uncertainty quantified via the 95\% confidence intervals is indicated by the dashed lines. Figure \ref{Fs5a:2} shows the distribution of obtained values for $\mathcal{R}_0$, the mean, 5th, and 95th percentiles being respectively 2.623, 2.042, and 3.240.
	
	\begin{figure}[H] 
		\centering
		{\includegraphics[scale=0.25]{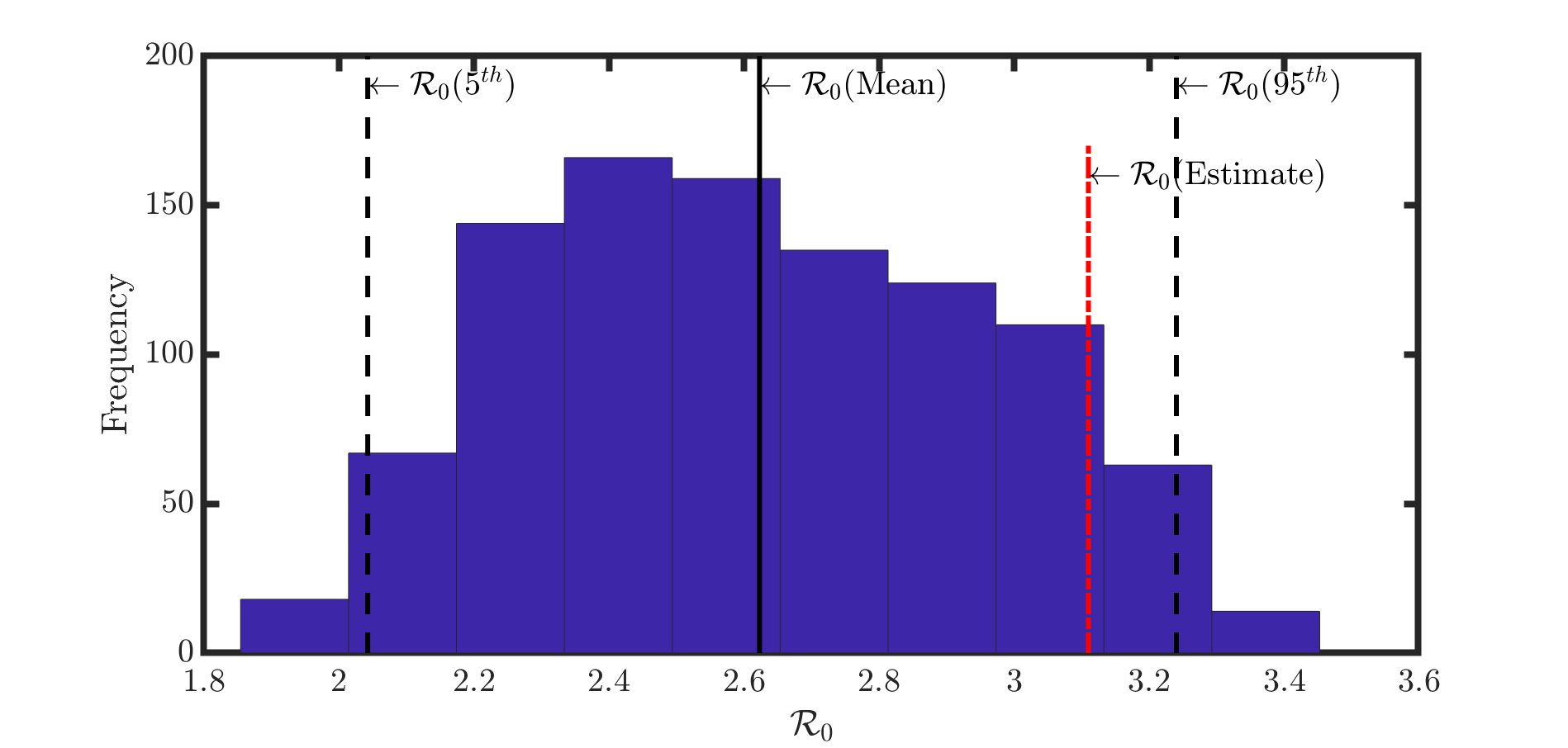}}
		\caption{
			{Uncertainty analysis of the basic reproduction number $ \mathcal{R}_{0} $ depicted by the histogram with plot showing 95\% confidence interval (dashed lines), mean (solid line) and an estimate (red dotted-dashed line) of $\mathcal{R}_{0}$ (\ref{Eqs3b:6})}.}
		\label{Fs5a:1}
	\end{figure}
	
	Using the best-fit values of all the parameters given in Table \ref{Ts4a:2} yields an estimate of $\mathcal{R}_0$ towards the upper end of the distribution, namely a value of 3.110 (see the red dotted-dashed line). In general, the higher the uncertainty, the wider the spread of the distribution of $\mathcal{R}_0$. We note that there is some uncertainty in $\mathcal{R}_0$ due to the model parameter estimates in Table \ref{Ts4a:2}; however, this is less than for most parameters.
	In Table \ref{Ts4a:2}, for almost all parameters, the upper and lower 95\% confidence intervals differ from the best fit value by a factor of two. However, for $\mathcal{R}_0$, the upper and lower ends of the interval are with $\pm$24\% of the mean value; thus overall, the uncertainty in the estimate of $\mathcal{R}_0$ is less than that of the individual parameters.
	
	\begin{figure}[H] 
		\hspace*{-0.15in}
		{\includegraphics[scale=0.25]{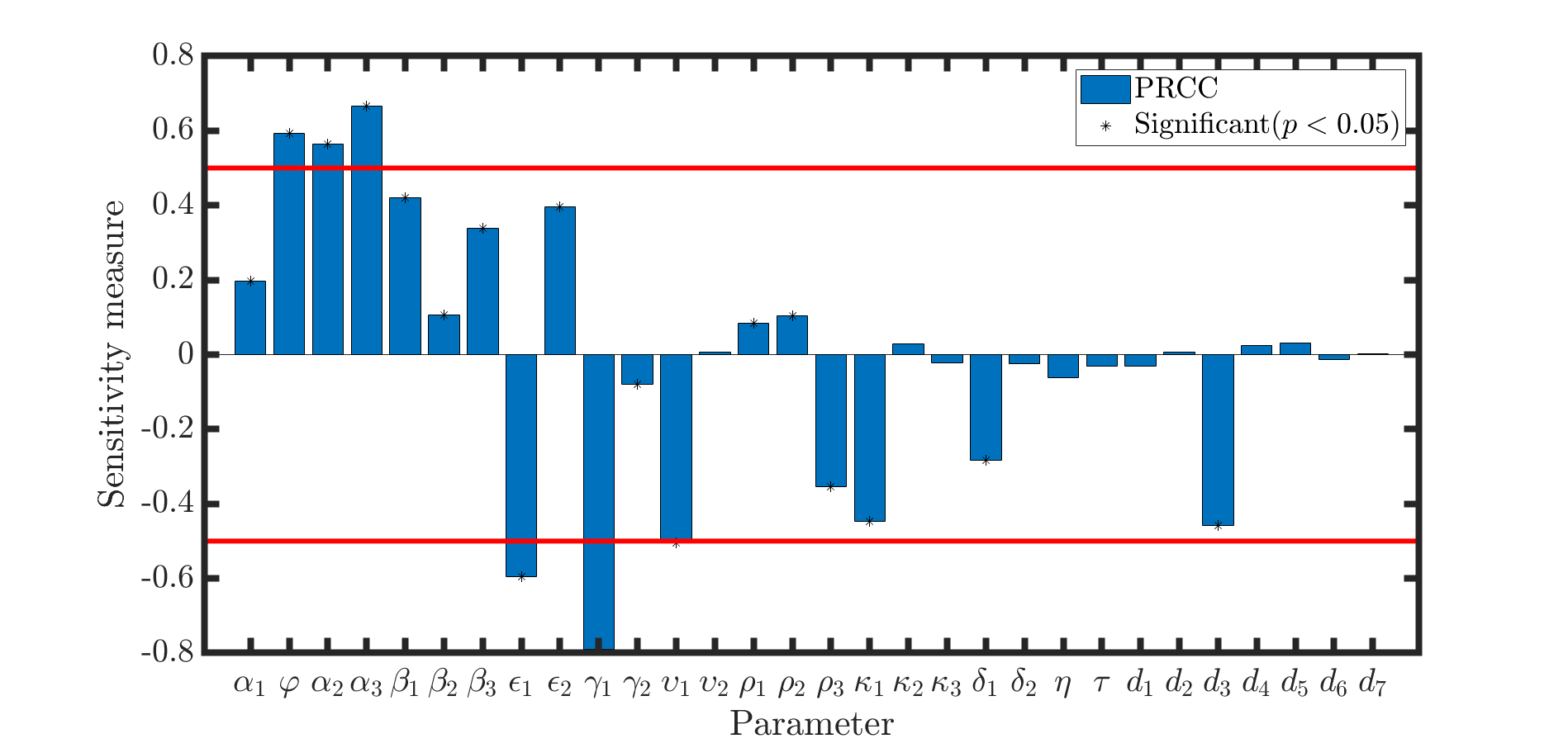}}
		\caption{
			{Sensitivity of the basic reproduction number $\mathcal{R}_{0} $ to changes in the CoVCom9 parameters using PRCC index}.}
		\label{Fs5a:2}
	\end{figure}
	
	Figure \ref{Fs5a:2} shows the sensitivity of the reproduction number $\mathcal{R}_0$ to each of the parameters in the underlying model (\ref{Eqs2:3}).
	PRCC assigns each parameter a value between $-1$ and $+1$. The magnitude of PRCC shows the parameter importance while the sign of PRCC gives the direction of the relationship between the input parameter and the model output of interest. Negative PRCC values mean that as the parameter value increases, the value of the model output of interest decreases and vice versa.
	The results of the PRCCs depicted help identify which parameters are primarily responsible for the uncertainty in $\mathcal{R}_0$, which suggests those interventions which should be most efficacious in controlling the spread of the virus by reducing $\mathcal{R}_{0}$.
	A PRCC value of zero gives an indication of no association between the input parameter and model output of interest.  The most significant model parameters are those associated with small $p-$values ($p < 0.05$) and large magnitude PRCC values ($ 0.5 \le |PRCC| \le 1$).
	
	From Figure \ref{Fs5a:2}, we identify six parameters as most influential
	on the basic reproduction number, $\mathcal{R}_{0}$, these are:
	\begin{itemize}
		\item $\varphi$ - the transmission rate of infectious individuals,
		\item $\alpha_{2}$ -  the probability of transmission of quarantine suspected infectious individuals,
		\item $\alpha_{3}$ - the probability of transmission of confirmed-positive infectious individuals,
		\item $\epsilon_{1}$ - the progression rate of exposed individuals to infectious class \item $\gamma_{1}$ - the progression rate of infectious individuals to confirmed-positive class, and
		\item $\upsilon_{1}$ - the progression rate of quarantine suspected infectives to the class of confirmed-positive cases.
	\end{itemize}
	In particular, $\mathcal{R}_{0}$ increases with increases in $\varphi$, $\alpha_{2}$ and $\alpha_{3}$, while $\mathcal{R}_{0}$ decreases with increases in $\epsilon_{1}$, $\gamma_{1}$ and $\upsilon_{1}$.  It is therefore critical that intervention strategies should be aimed at decreasing the values of $\varphi$, $\alpha_{2}$ and $\alpha_{3}$ and increasing the values of $\epsilon_{1}$, $\gamma_{1}$ and $\upsilon_{1}$.
	
	These recommendations should not be interpreted as discounting the value of considering efforts to alter other significant model parameters such as probability of transmission of hospitalised individuals at ordinary ward ($\beta_{1}$), the progression rate of exposed individuals to quarantine suspected exposed class ($\epsilon_{2}$), and recovery rate of hospitalised individual ($\kappa_{1}$).
	
	\subsection{Predicting the effects of lockdown}
	\label{S:ns}
	
	The simulation presented in Figure \ref{Fs4a:2} show a worrying trend 
	of exponential growth with no sign of plateau or reduction in the 
	effects of the pandemic.  Many countries have implemented a `lockdown', 
	that is regulations to restrict social interactions and so reduce 
	the spread of the disease.  Here, we model the effects of lockdown 
	by a simple reduction in the parameter $\varphi$, and simulate 
	the spread by solving the model using the standard value of $\varphi$ 
	for the first 350 days, and a lower value of $\varphi$
	for the time period $350 \leq t \leq 700$ days.   The results are 
	presented in Figures \ref{Fsns:6}, \ref{Fsns:7}, \ref{Fsns:8}, 
	for the values $\varphi = 0.008$, 0.016, 0.004, 
	the first value of $\varphi$ being chosen so as to reduce the 
	expected value of $\mathcal{R}_0$ from 3.110 to 0.995, the threshold 
	required for containment of the epidemic. The second and third 
	values are chosen to be double and half of this critical value. 
	Note that the vertical scales in Figures \ref{Fsns:6}, 
	\ref{Fsns:7}, and \ref{Fsns:8} are not identical. 
	
	\begin{figure}[H] 
		\hspace*{-0.60in}
		\centering
		\includegraphics[scale=0.28]{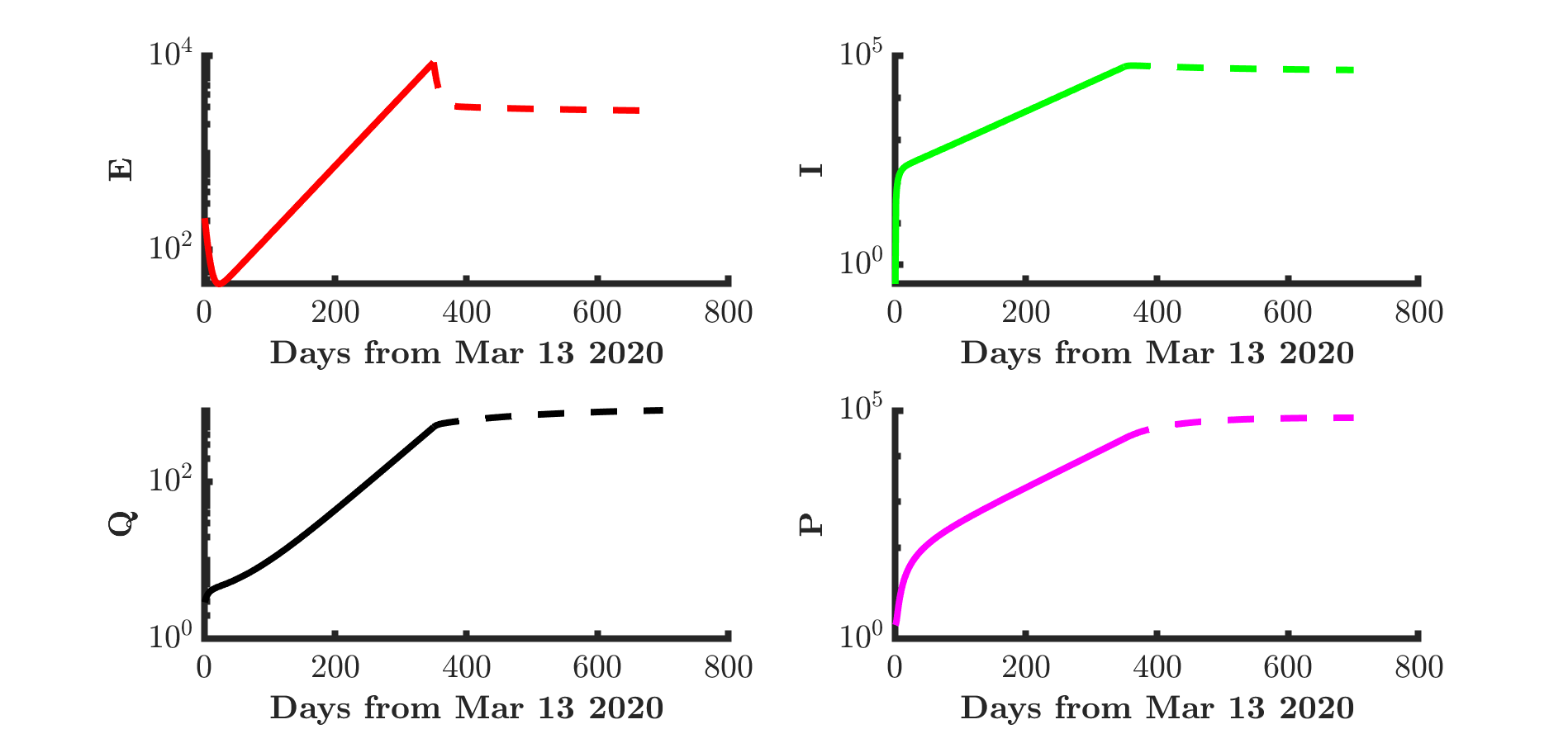}\\
		\hspace*{-0.6in}
		\includegraphics[scale=0.28]{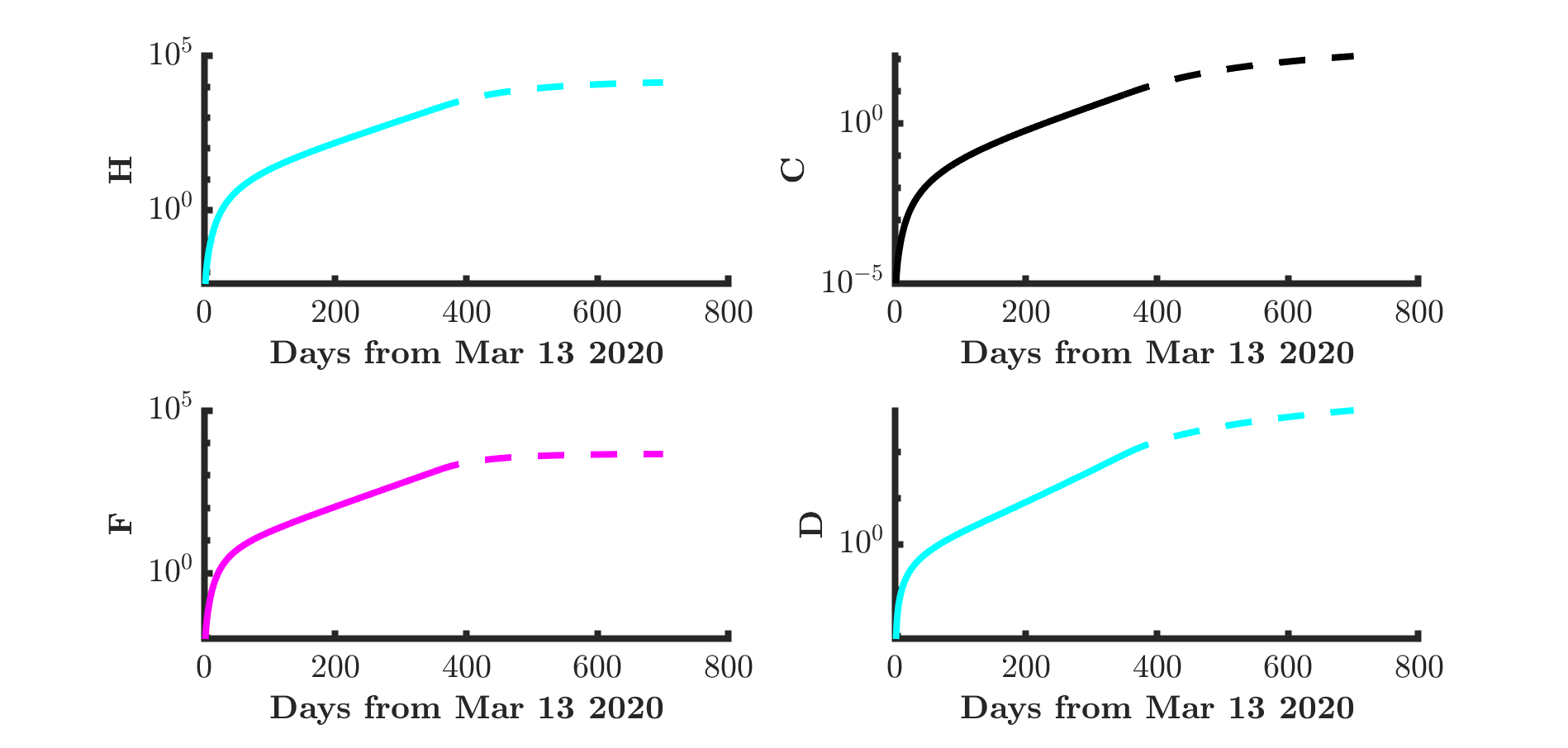}
		\caption{\textit{One-year simulation dynamics of CoVCom9 model from March 13 2020 when there is a 68\% reduction in $\varphi$, that is, to $\varphi=0.008$. Here $E$, $I$, $Q$, $P$, $H$, $C$, $F$ and $D$ are respectively exposed, infectious, quarantine suspected-expose, confirmed-positive, hospitalised at ordinary ward, hospitalised at intensive care unit, and deaths with the vertical axis on a log-scale. $\mathcal{R}_{0}$ changes from 3.110 to 0.995}.}
		\label{Fsns:6}
	\end{figure}
	
	Figure \ref{Fsns:6} shows a clear almost instant reduction 
	in the number of exposed people ($E$), followed by a plateau, 
	whilst the sizes of most other sub-populations plateau. 
	However, the numbers of hospitalised cases ($H$ and $C$) 
	both continue to rise slowly.  We see that this strength of 
	lockdown stops the exponential growth.  
	The less severe lockdown simulated in Figure \ref{Fsns:7} 
	causes a brief reduction in the number of exposed cases; 
	however, the exponential growth is quickly resumed, 
	in the size of all sub-populations, 
	albeit with a slightly smaller growth rate. 
	The more severe lockdown simulated in Figure \ref{Fsns:8} 
	shows a sudden and sharp reduction in the number of exposed 
	($E$), followed by a steady exponential decrease.  
	The numbers of infected, quarantined and positive cases 
	is also seen to fall exponentially, whilst the cases 
	of hospitalised, intensive care, and self-isolated
	all plateau, as the total number of deaths slowly increases.  
	
	\begin{figure}[H] 
		\hspace*{-0.60in}
		\centering
		\includegraphics[scale=0.28]{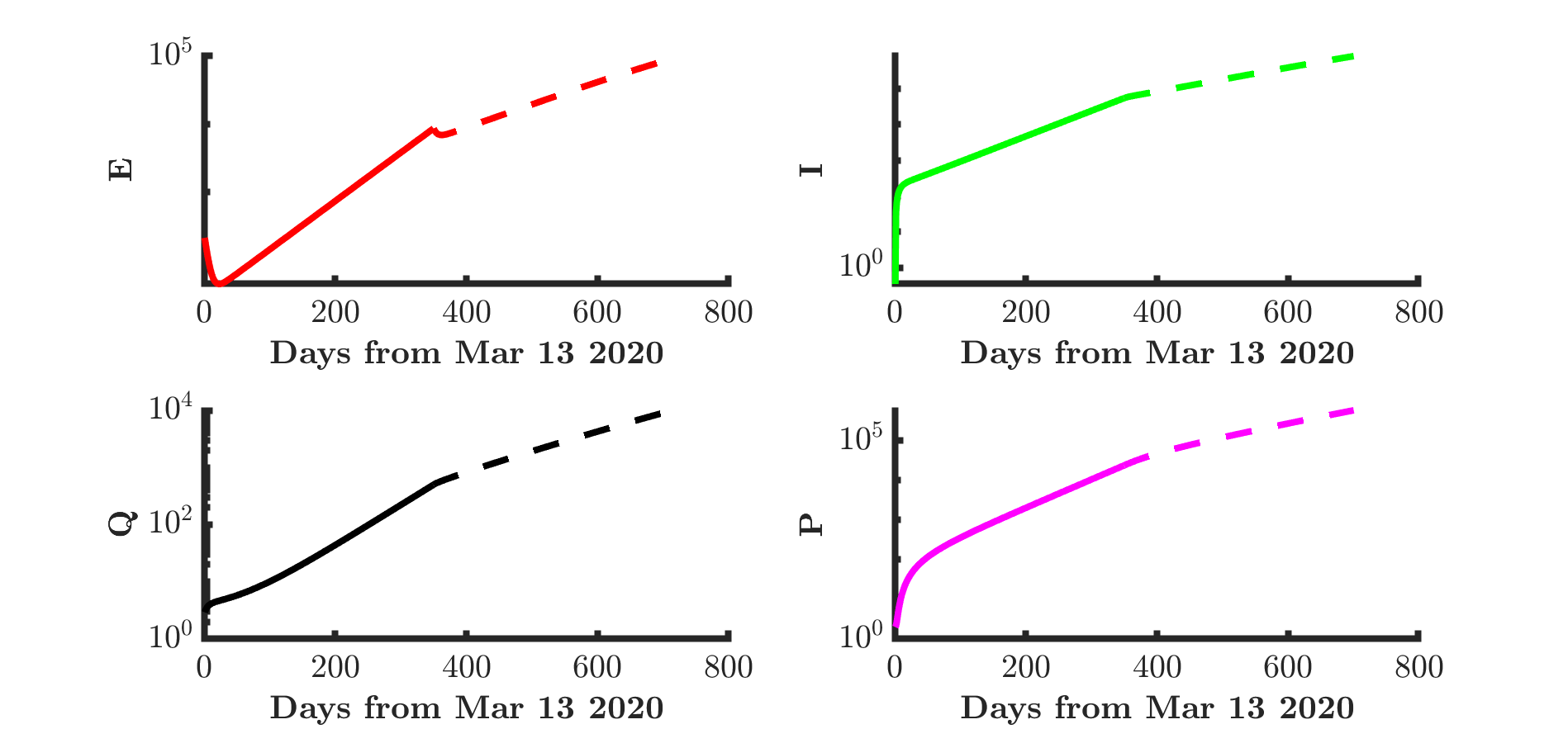}\\
		\hspace*{-0.6in}
		\includegraphics[scale=0.28]{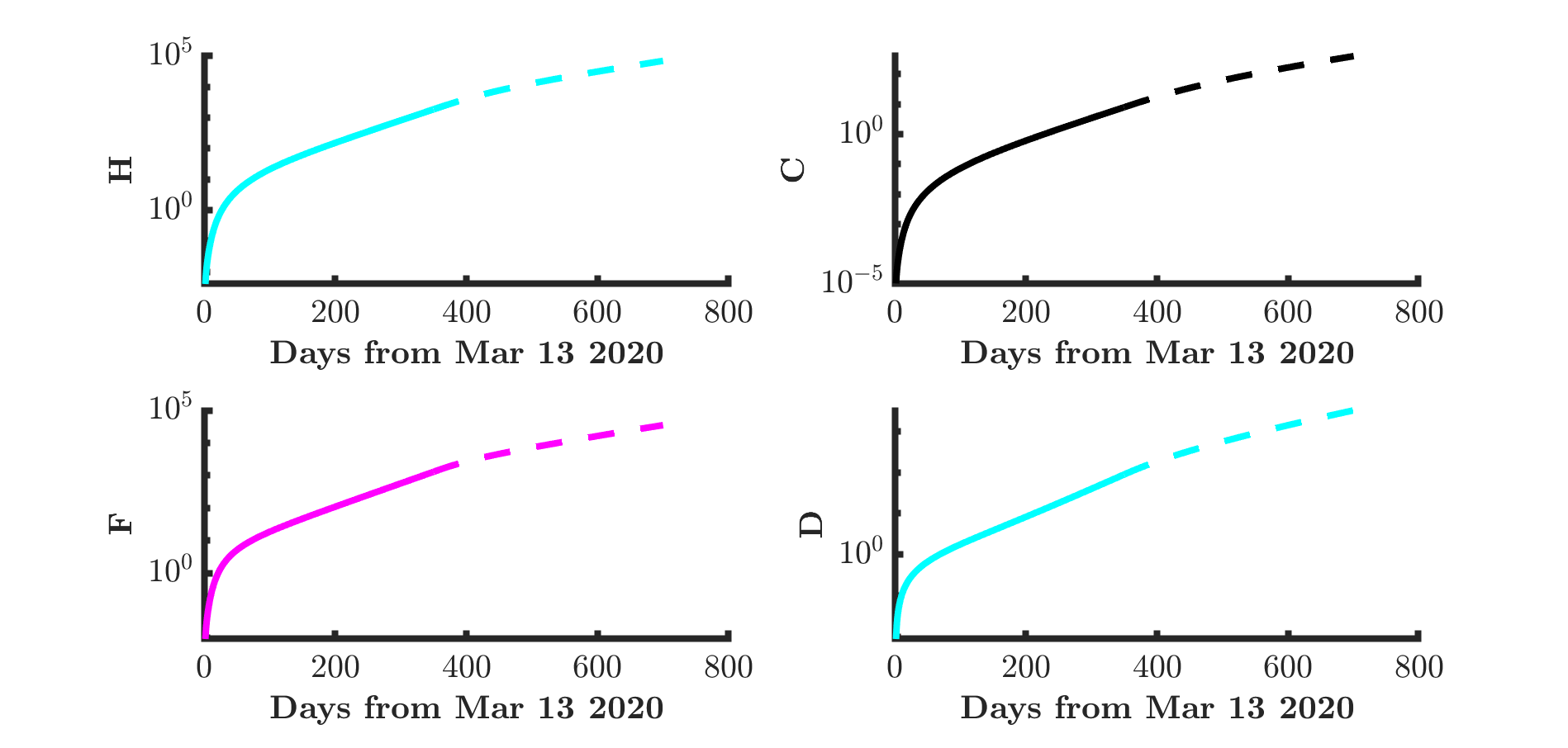}
		\caption{\textit{One-year simulation dynamics of CoVCom9 model from March 13 2020 when there is a 34\% reduction in $\varphi$, that is, to $\varphi=0.016$. Again, $E$, $I$, $Q$, $P$, $H$, $C$, $F$ and $D$ are respectively exposed, infectious, quarantine suspected-expose, confirmed-positive, hospitalised at ordinary ward, hospitalised at intensive care unit, and deaths with the vertical axis on a log-scale. The effect on $\mathcal{R}_{0}$ is a change from 3.110 to 2.053}.}
		\label{Fsns:7}
	\end{figure}
	
	It should be noted that these simulations are only a crude model 
	of the effects of lockdown, in reality a lockdown could cause 
	changes to other parameters, particularly $\alpha_1$, $\alpha_2$, 
	$\alpha_3$, $\beta_1$, $\beta_2$, $\beta_3$, in the formula 
	(\ref{Eqs2:2}) for the spread of the disease. We leave the topic 
	of more detailed models of the effects of lockdown for future work. 
	
	\begin{figure}[H] 
		\hspace*{-0.60in}
		\centering
		\includegraphics[scale=0.28]{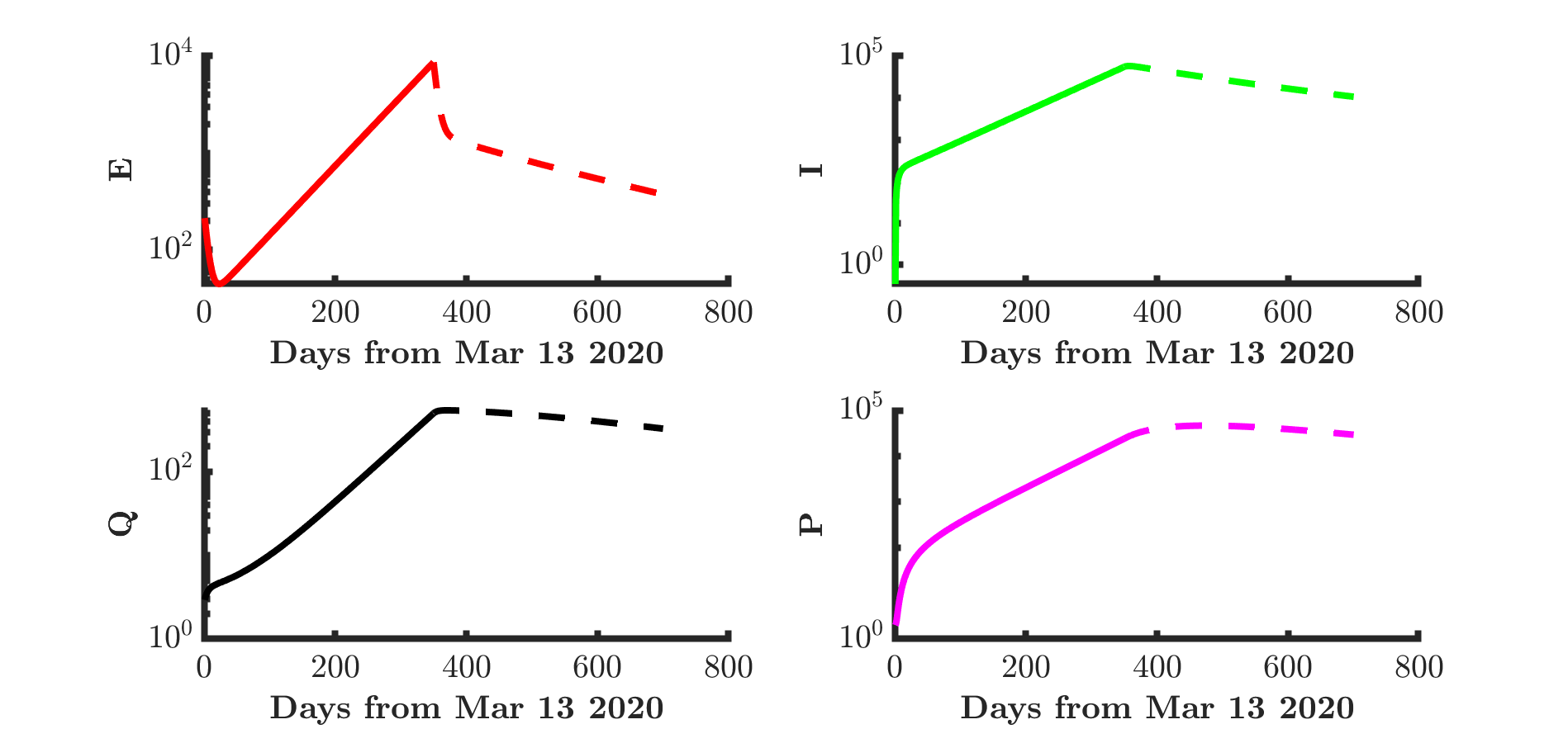}\\
		\hspace*{-0.6in}
		\includegraphics[scale=0.28]{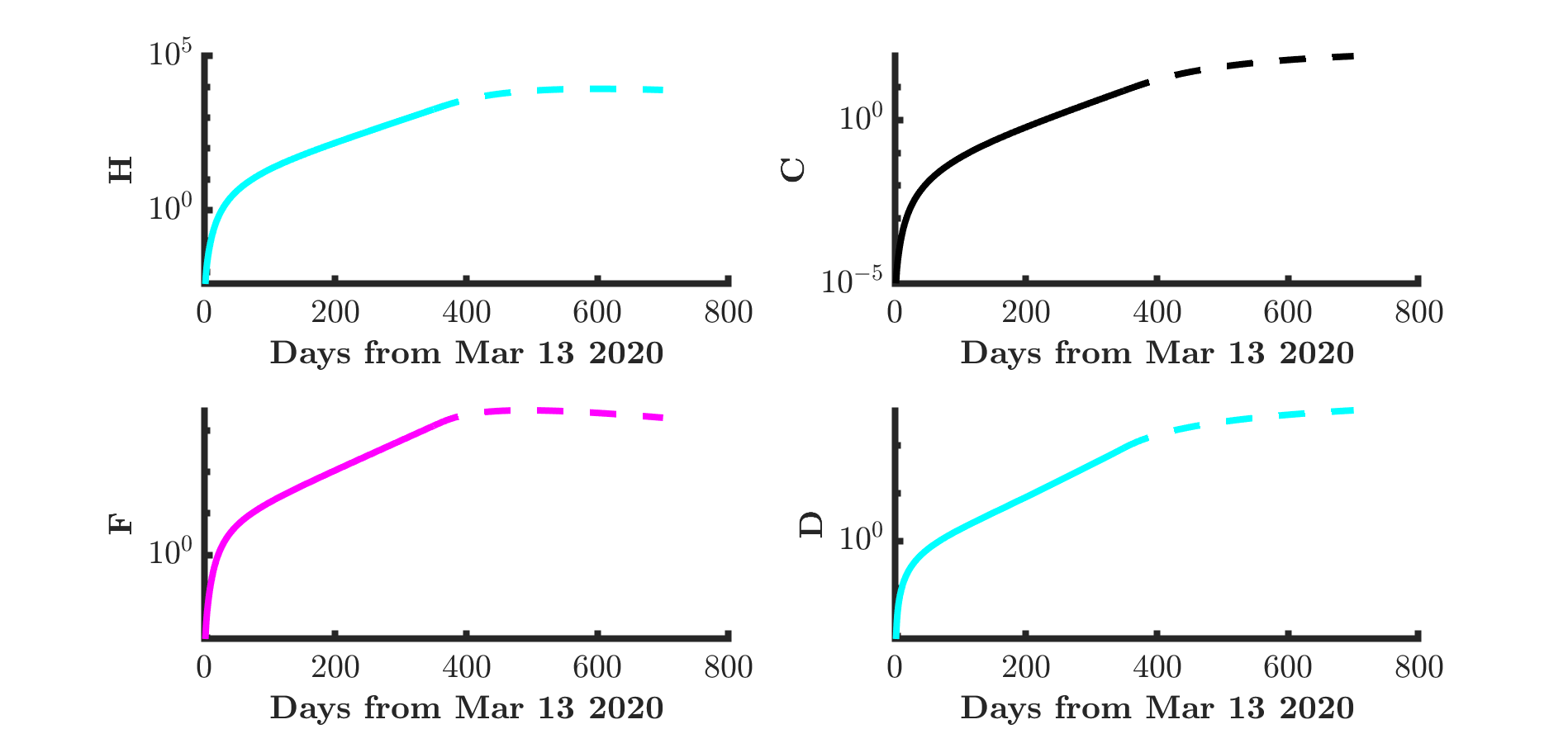}
		\caption{\textit{One-year simulation dynamics of CoVCom9 model from March 13 2020 when there is a 84\% reduction in $\varphi$, that it, to $\varphi=0.004$. As above, $E$, $I$, $Q$, $P$, $H$, $C$, $F$ and $D$ are respectively exposed, infectious, quarantine suspected-expose, confirmed-positive, hospitalised at ordinary ward, hospitalised at intensive care unit, and deaths with the vertical axis on a log-scale. The effect on $\mathcal{R}_{0}$ is a change from 3.110 to 0.498}.}
		\label{Fsns:8}
	\end{figure}
	
	
	\section{Discussion and conclusions}
	\label{S:6}
	
	We have developed a mathematical model (CoVCom9) in the form of a system of coupled ordinary differential equations to describe SARS-CoV-2 transmission dynamics in Ghana.  This categorises every member of the population into one of 9 classes, including various classes {\textit{well-defined and measurable classes}}, such as those who have tested positive for SARS-Cov-2 and are hospitalised (ordinary wards/intensive care), quarantined, etc, as well as {\textit{unmeasurable but clinically important classes}}, such as those who have been exposed to the virus, those who are infectious but not yet tested positive.  We investigated the epidemiological well-posedness of the CoVCom9 model, shown that solutions remain positive, and analysed the stability of the equilibrium solution. Using a candidate Lyapunov function, we have shown that the disease-free equilibrium is globally asymptotically stable when the basic reproduction number is $\mathcal{R}_{0}<1$.
	
	Using the reported data \citep{ritchie2020coronavirus} from March 13, 2020, to August 10, 2020, for both confirmed-positive cases and deaths of SARS-CoV-2 disease, we have parameterised the CoVCom9 model, with other parameters being estimated based from the literature. During the parameter estimation exercise, we used system exploratory analysis (SEA) to find practical parameter spaces. The estimated parameter values provided best fits that are in good agreement with both reported confirmed-positive cases and deaths. Also, the results point that on March 13, 2020, while two individuals are confirmed-positive, approximately 213, 3, and 1 persons were respectively exposed, quarantine suspected and infectious.
	
	We have used Latin Hypercube Sampling-Rank Correlation Coefficient (LHS-PRCC) to investigate the uncertainty and sensitivity of the reproduction number $\mathcal{R}_0$.  The results derived are of significant epidemiological value in SARS-CoV-2 control. We estimate that over the period, March-August 2020, the average basic reproduction number for Ghana was $\mathcal{R}_{0}=3.110$, which has the 95\% confidence percentile interval (2.042 - 3.240, in approximate centre of this interval is the mean value of 2.623). From Figure \ref{Fs5a:2}, we note that $\mathcal{R}_{0}$ is most sensitive to six model parameters ($\varphi$, $\alpha_{2}$, $\alpha_{3}$, $\epsilon_{1}$, $\gamma_{1}$, and $\upsilon_{1}$ whose effects are detailed in Table \ref{Ts2a:2}).
	
	The proposed CoVCom9 model is a result of our effort to gain insight into the vital features of SARS-CoV-2 transmission dynamics in Ghana. Future work will be focused on extending the model to account for inflow into other classes due to opening of Ghana's borders. Further, we will consider time-dependent optimal control intervention strategies to gain insight into the best strategy for Ghana. Other extensions include the time-dependent force of infection and the maximum capacity of intensive care units.

	
	\section*{Declaration of competing interest}
	
	The authors declare that they have no known competing financial interests or personal relationships that could have appeared to influence the study reported in this paper.
	
	\section*{Acknowledgements}
	
	The authors (EA, JADW and RLG) are thankful for funding provided by the Leverhulme Trust Doctoral Scholarship (DA214-024), Modelling and Analytics for a Sustainable Society (MASS), and to the University of Nottingham.

	%
	
	\appendix
	\section{Appendix}
	\label{S:7}
	
	\setcounter{equation}{0}
	\setcounter{figure}{0}
	\setcounter{table}{0}
	\renewcommand{\thefigure}{\Alph{section}.\arabic{figure}}
	\renewcommand{\thetable}{\Alph{section}.\arabic{table}}
	\renewcommand\theequation{\Alph{section}.\arabic{equation}}

	In Table \ref{Ts7:2} we list the parameter values 
	used in the simulations presented in Section \ref{S:4}. 
	
	Figures \ref{Fsns:A1}, \ref{Fsns:A2}, \ref{Fsns:A3} show our predictions 
	for how the subpopulation sizes in the model would have evolved over time 
	if a lockdown had been imposed as soon as the first cases entered Ghana. 
	These predictions are obtained by keeping all parameters at the same values 
	as in the main model, and reducing $\varphi$ to the values used 
	in Section \ref{S:ns}. 
	These graphs should be compared with Figure \ref{Fs4a:2}. 
	In Figure \ref{Fsns:A1} we use $\varphi=0.008$, 
	which is chosen to make our estimate of $\mathcal{R}_0=1$. 
	We see that this has the effect of bringing the pandemic under 
	some sort of control, but only over an extremely long timescale. 
	In Figure \ref{Fsns:A2} we simulate a partial lockdown, that is, 
	reducing $\varphi$ to 0.016 - which is the midpoint of the standard 
	value $\varphi=0.02495$ and that required to reduce $\mathcal{R}_0$ to 1.  
	We see that epidemic still grows, but at a slower rate than with no lockdown. 
	Finally, in Figure \ref{Fsns:A3}, we consider the effect of 
	a much more severe lockdown, where $\varphi$ is reduced to half 
	that needed for $\mathcal{R}_0=1$, that is $\varphi=0.004$.  This suggests 
	that the epidemic can be controlled and eliminated within a year. 
	
	\begin{table}[H] 
		\hspace*{0.05in}
		\centering
		\begin{minipage}{0.9\textwidth}
			\fontsize{8}{10}\selectfont
			\caption{Estimated initial values model variables and parameters for the system (\ref{Eqs2:2}.)}
			\label{Ts7:2}
			\centering
			\begin{tabular}{p{1cm} p{3cm} p{3cm}}
				\hline \hline
				\multicolumn{1}{l}{Parameters}
				& \multicolumn{1}{l}{Min}
				& \multicolumn{1}{l}{Max}\\
				\hline
				$ \alpha_{1} $
				& 0.245545
				& 1\\
				$ \varphi $
				& 0.012654
				& 0.050619\\
				$ \alpha_{2} $
				& 0.210646
				& 0.842586\\
				$ \alpha_{3} $
				& 0.020245
				& 0.080978\\
				$ \beta_{1} $
				& 0.089411
				& 0.357644\\
				$ \beta_{2} $
				& 0.060902
				& 0.243607\\
				$ \beta_{3} $
				& 0.169305
				& 0.677222\\
				$ \epsilon_{1} $
				& 1/7
				& 1/3 (a)\\
				$ \epsilon_{2} $
				& 0.000873
				& 0.003492\\
				$ \gamma_{1} $
				& 0.008402
				& 0.033610\\
				$ \gamma_{2} $
				& 8.662695$\times 10^{-5}$
				& 0.000347\\
				$ \upsilon_{1} $
				& 0.000379
				& 0.001517\\
				$ \upsilon_{2} $
				& 6.863871$\times 10^{-10}$
				& 2.745548$\times 10^{-9}$\\
				$ \rho_{1} $
				& 0.000547
				& 0.002188\\
				$ \rho_{2} $
				& 2.477852$\times 10^{-6}$
				& 9.911408$\times 10^{-6}$\\
				$ \rho_{3} $
				& 0.001391
				& 0.005565\\
				$ \kappa_{1} $
				& 0.005728
				& 0.022913\\
				$ \kappa_{2} $
				& 2.874618$\times 10^{-6}$
				& 1.149847$\times 10^{-5}$\\
				$ \kappa_{3} $
				& 1.488217$\times 10^{-5}$
				& 5.952870$\times 10^{-5}$\\
				$ \delta_{1} $
				& 1/23
				& 1/11 (b)\\
				$ \delta_{2} $
				& 2.862673$\times 10^{-9}$
				& 1.145069$\times 10^{-8}$\\
				$ \eta $
				& 4.861253$\times 10^{-5}$
				& 0.000194\\
				$ \tau $
				& 7.740808$\times 10^{-9}$
				& 3.096323$\times 10^{-8}$\\
				$ d_{1} $
				& 3.892588$\times 10^{-10}$
				& 1.557035$\times 10^{-9}$\\
				$ d_{2} $
				& 6.221547$\times 10^{-14}$
				& 2.488619$\times 10^{-13}$\\
				$ d_{3} $
				& 0.000996
				& 0.003985\\
				$ d_{4} $
				& 2.996889$\times 10^{-10}$
				& 1.198756$\times 10^{-9}$\\
				$ d_{5} $
				& 6.838531$\times 10^{-13}$
				& 2.735412$\times 10^{-12}$\\
				$ d_{6} $
				& 3.412952$\times 10^{-14}$
				& 1.365181$\times 10^{-13}$\\
				$ d_{7} $
				& 1.200917$\times 10^{-12}$
				& 4.803667$\times 10^{-12}$\\
				$ E $
				& 40
				& 300\\
				$ I $
				& 0
				& 10\\
				$ Q $
				& 1
				& 70\\
				\hline \hline
			\end{tabular}
			\\ (a) denotes literature values and (b) denotes assumed value
		\end{minipage}
	\end{table}
	
	\begin{figure}[H] 
		\hspace*{-0.60in}
		\centering
		\includegraphics[scale=0.28]{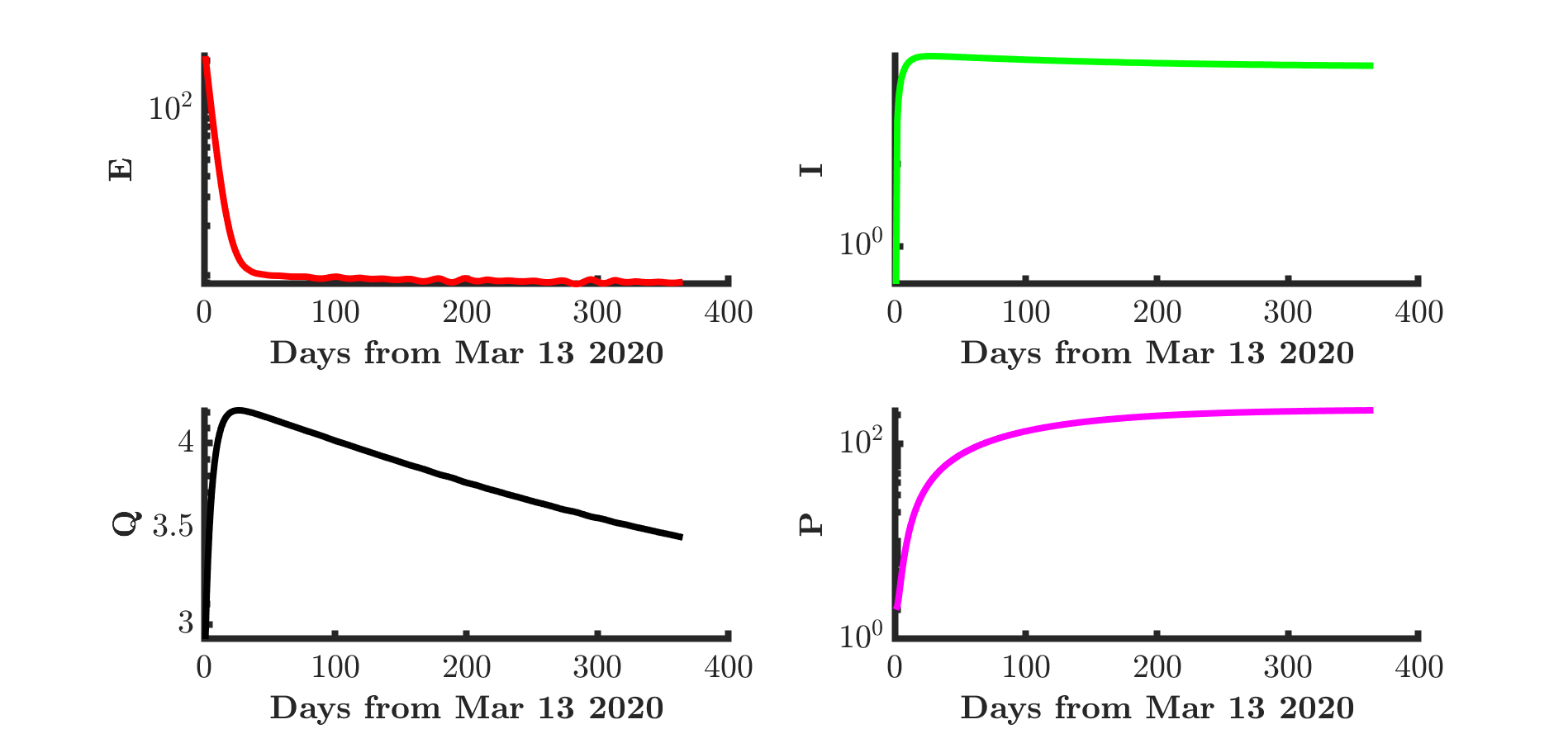}\\
		\hspace*{-0.6in}
		\includegraphics[scale=0.28]{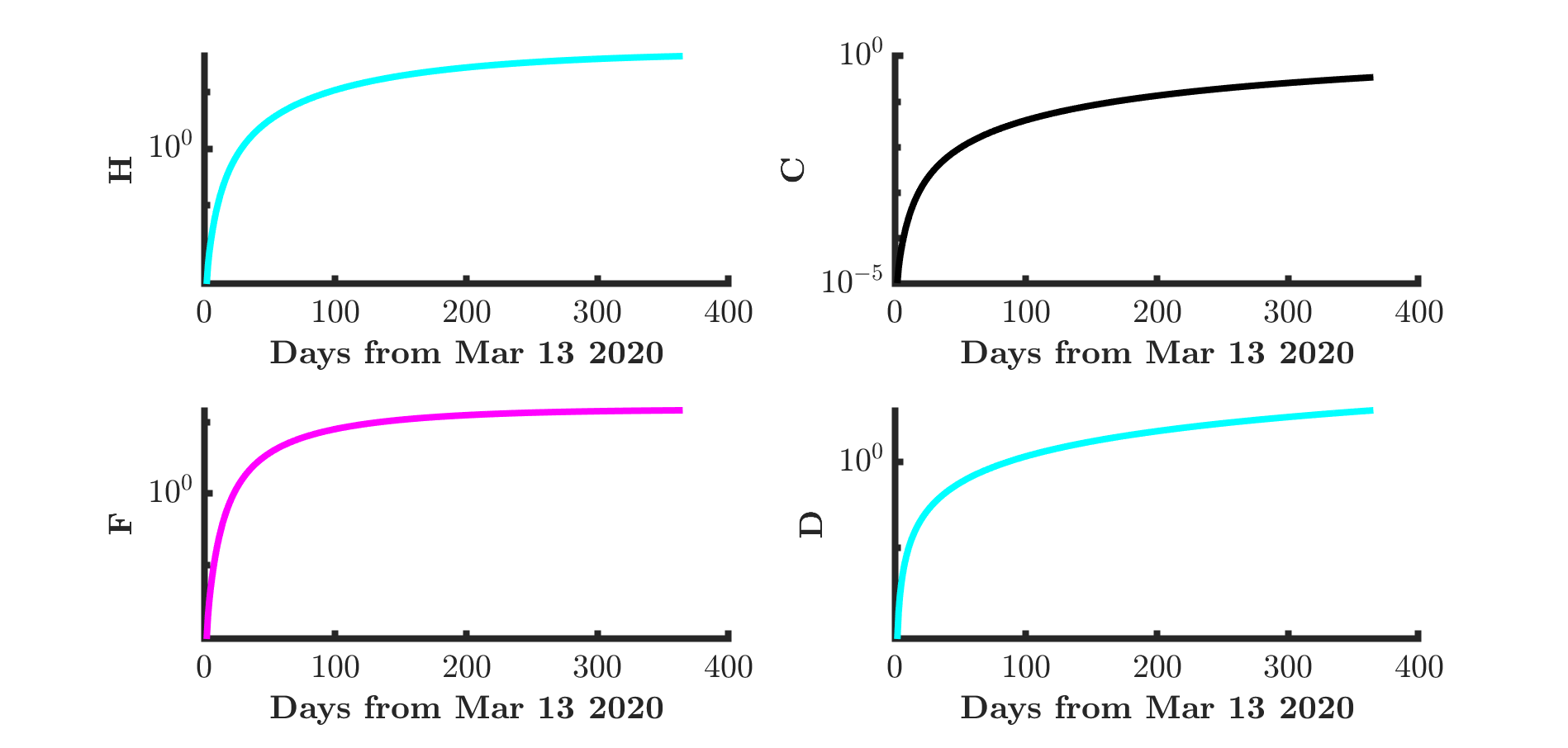}
		\caption{\textit{One-year simulation dynamics of CoVCom9 model from March 13 2020 when there is a 68\% reduction in $\varphi$ to $\varphi=0.008$. Here $E$, $I$, $Q$, $P$, $H$, $C$, $F$ and $D$ are respectively exposed, infectious, quarantine suspected-expose, confirmed-positive, hospitalised at ordinary ward, hospitalised at intensive care unit, and deaths with the vertical axis on a log-scale. The effect on $\mathcal{R}_{0}$ is a change from 3.110 to 0.995}.}
		\label{Fsns:A1}
	\end{figure}
	
	\begin{figure}[H] 
		\hspace*{-0.60in}
		\centering
		\includegraphics[scale=0.28]{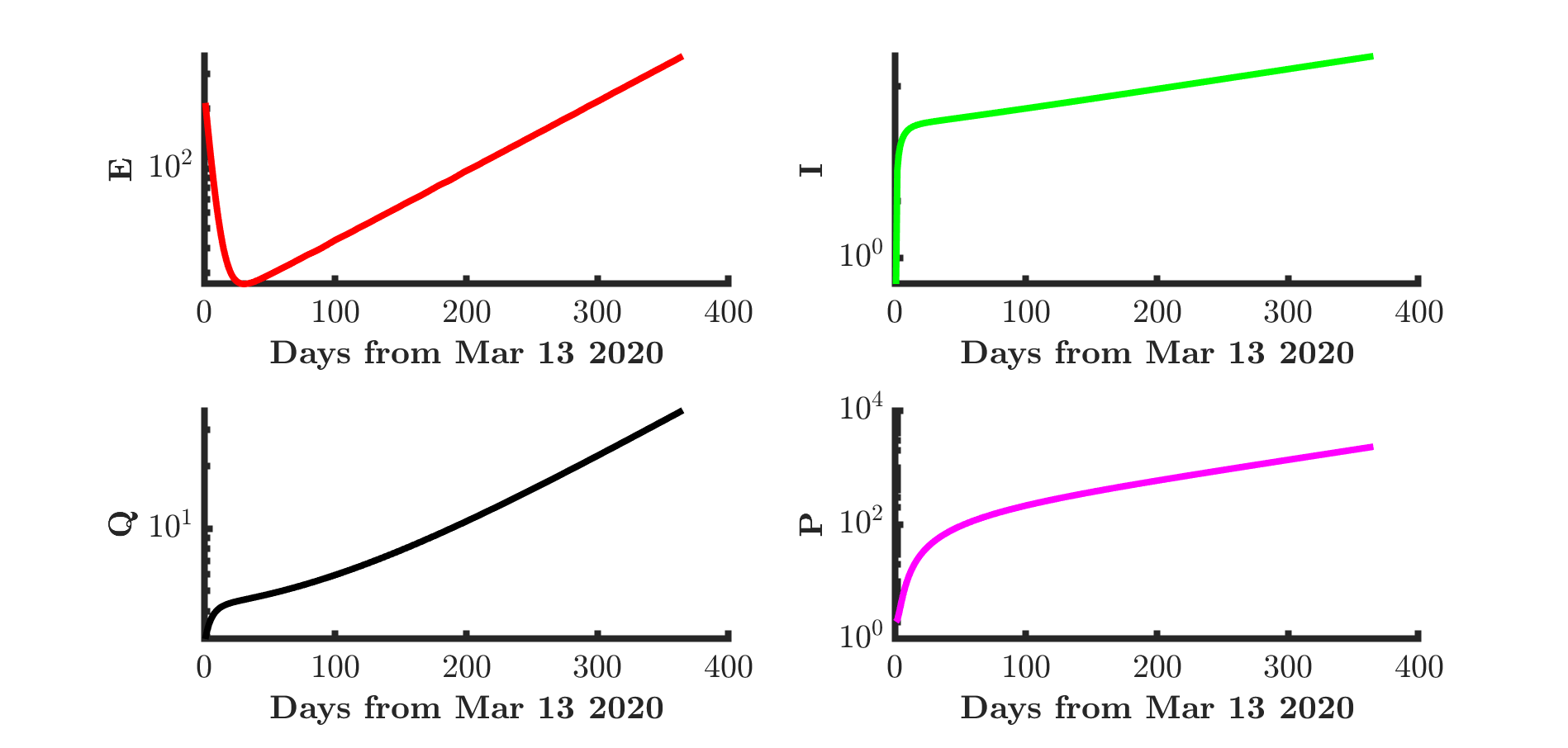}\\
		\hspace*{-0.6in}
		\includegraphics[scale=0.28]{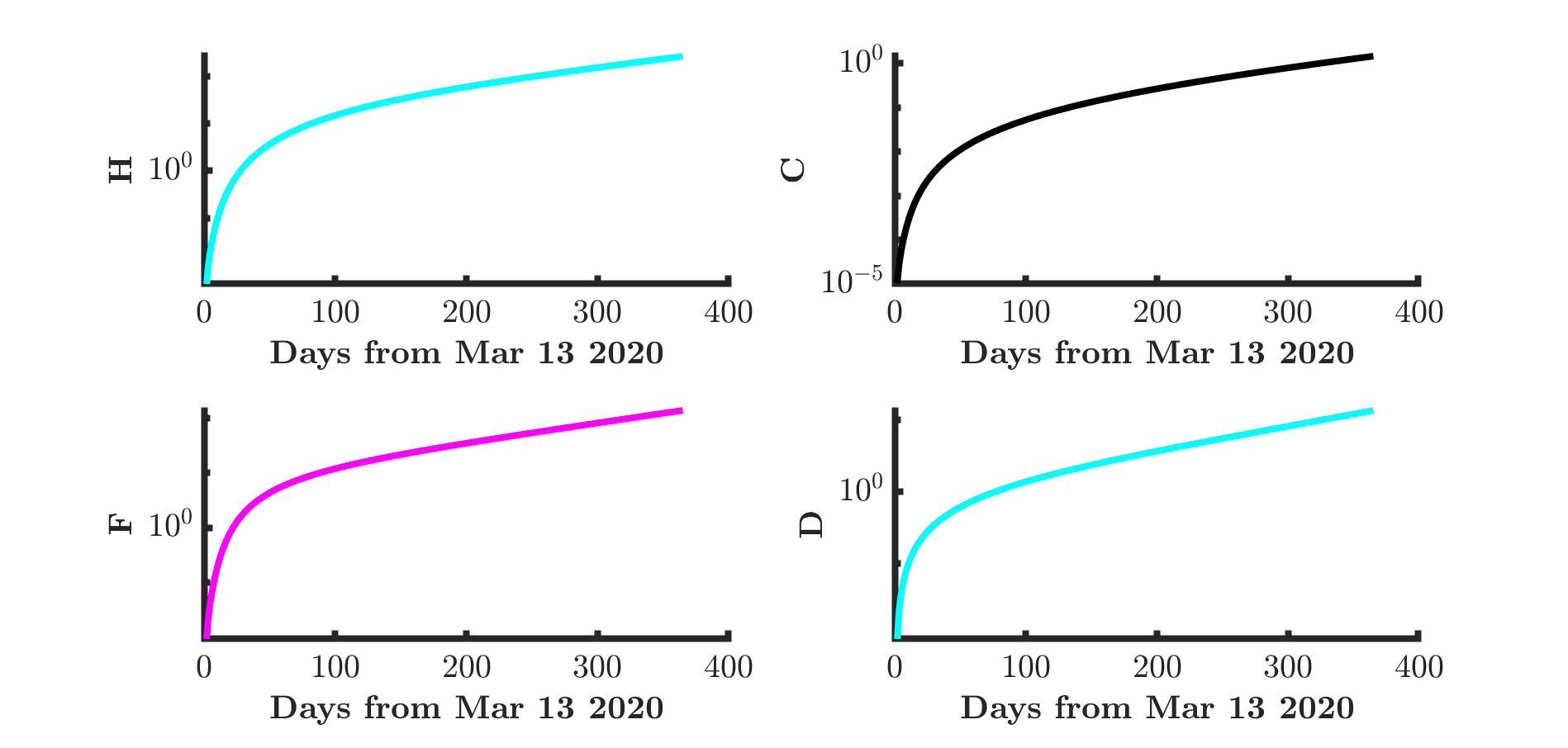}
		\caption{\textit{One-year simulation dynamics of CoVCom9 model from March 13 2020 when there is a 34\% reduction in $\varphi$ to $\varphi=0.016$. Here $E$, $I$, $Q$, $P$, $H$, $C$, $F$ and $D$ are respectively exposed, infectious, quarantine suspected-expose, confirmed-positive, hospitalised at ordinary ward, hospitalised at intensive care unit, and deaths with the vertical axis on a log-scale. The effect on $\mathcal{R}_{0}$ is a change from 3.110 to 2.053}.}
		\label{Fsns:A2}
	\end{figure}
	
	\begin{figure}[H] 
		\hspace*{-0.60in}
		\centering
		\includegraphics[scale=0.28]{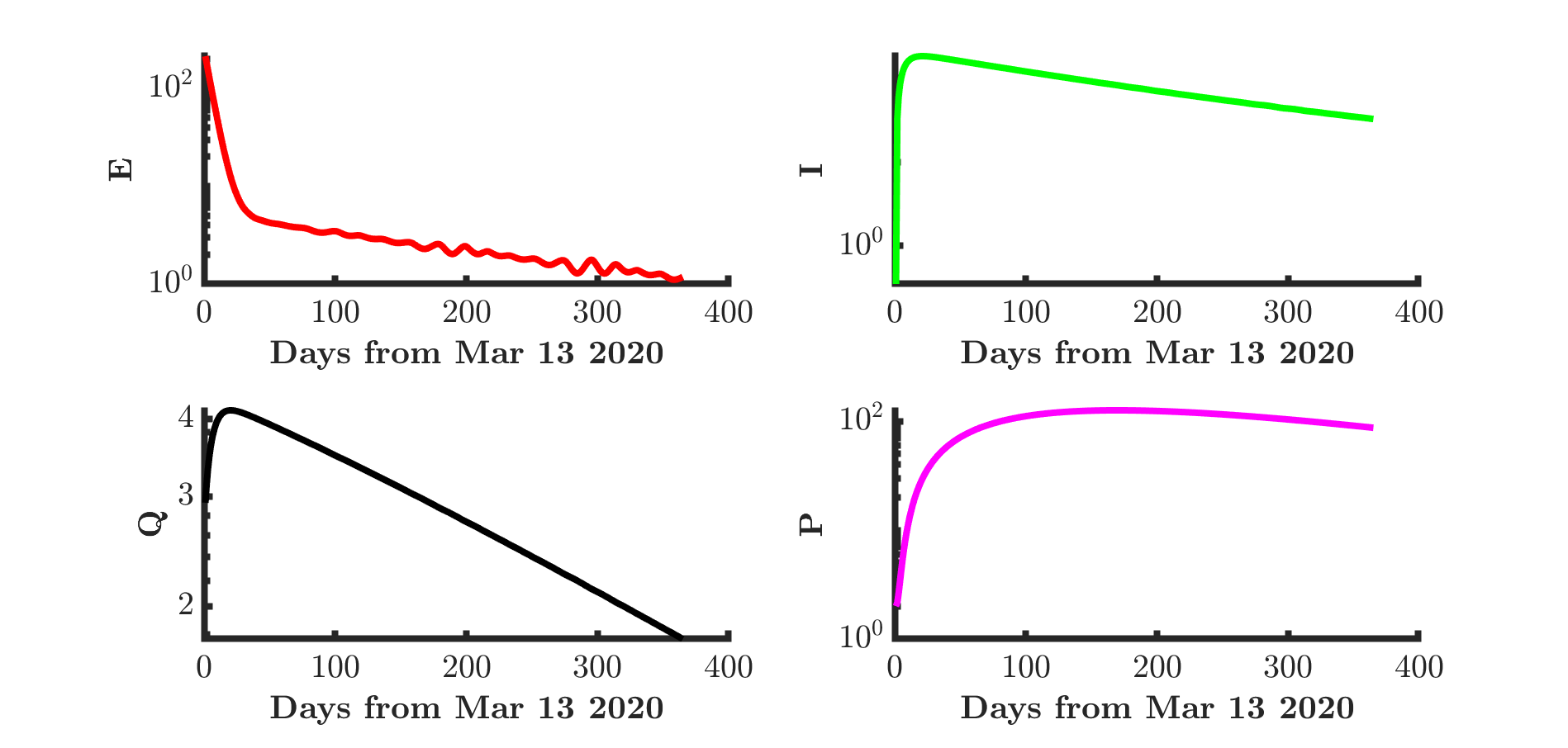}\\
		\hspace*{-0.6in}
		\includegraphics[scale=0.28]{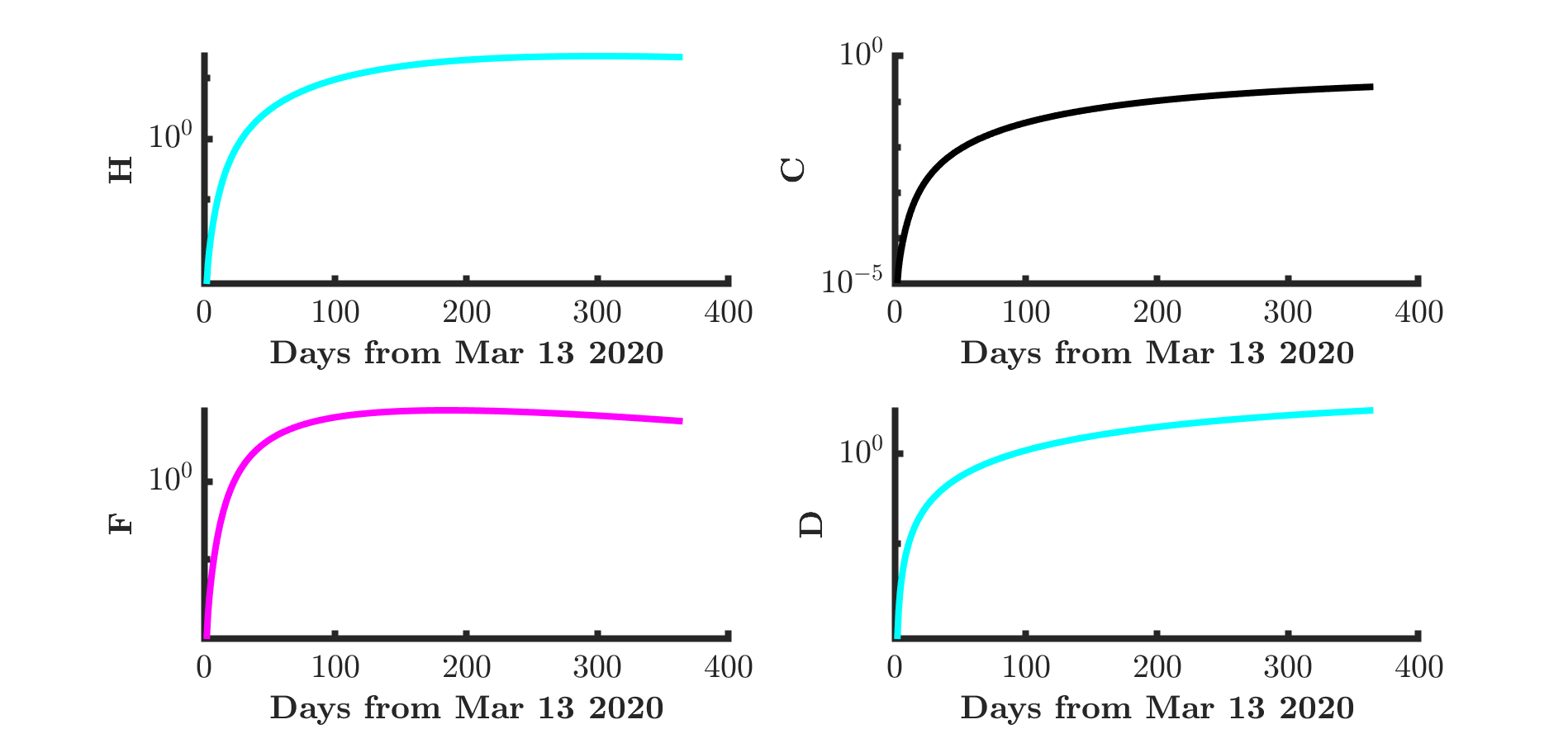}
		\caption{\textit{One-year simulation dynamics of CoVCom9 model from March 13 2020 when there is a 84\% reduction in $\varphi$, to $\varphi=0.004$. Here $E$, $I$, $Q$, $P$, $H$, $C$, $F$ and $D$ are respectively exposed, infectious, quarantine suspected-expose, confirmed-positive, hospitalised at ordinary ward, hospitalised at intensive care unit, and deaths with the vertical axis on a log-scale. The effect on $\mathcal{R}_{0}$ is a change from 3.110 to 0.498}.}
		\label{Fsns:A3}
	\end{figure}
	
	\newpage

	\bibliographystyle{elsarticle-harv}
	\bibliography{EA_covid19_rf}
	
\end{document}